\documentclass[11pt,letterpaper]{article}
\usepackage{fullpage}
\usepackage[round]{natbib}
\usepackage{amsmath,amsfonts,amssymb,amsthm, caption, comment, mathrsfs,mathtools}
\usepackage[margin=1in]{geometry}
\usepackage[most]{tcolorbox}
\usepackage{xcolor}
\usepackage[hyperindex,breaklinks]{hyperref}
\hypersetup{
  colorlinks,
  citecolor=violet,
  linkcolor=blue,
  urlcolor=blue}
\usepackage{dsfont}
\numberwithin{equation}{section}
\counterwithin{equation}{subsection}
\usepackage{booktabs} % For formal tables
\usepackage[ruled,noend]{algorithm2e} % For algorithms

\SetAlFnt{\small}
\SetAlCapFnt{\small}
\SetAlCapNameFnt{\small}
\SetAlCapHSkip{0pt}
\IncMargin{-\parindent}
\usepackage[colorinlistoftodos]{todonotes}
\numberwithin{equation}{section}
\usepackage{packages/macros-common}
\usepackage{packages/macros-project-specific}
\usepackage{packages/mytheorems}
\usepackage{lipsum}

\usepackage[suppress]{color-edits}
\addauthor{ss}{magenta} % ss for Suho
\addauthor{gp}{blue}
\addauthor{ih}{red}
\addauthor{mh}{magenta}

\newcommand\extrafootertext[1]{%
    \bgroup
    \renewcommand\thefootnote{\fnsymbol{footnote}}%
    \renewcommand\thempfootnote{\fnsymbol{mpfootnote}}%
    \footnotetext[0]{#1}%
    \egroup
}

\title{Gains-from-Trade in Bilateral Trade with a Broker}

  \newcommand{\country}[1]{#1.}
  \newcommand{\city}[1]{#1}
  \newcommand{\institution}[1]{#1}
  \newcommand{\email}[1]{Email: \texttt{#1}}
  \newcommand{\affiliation}{\thanks}

\author{
 Ilya Hajiaghayi\affiliation{
   \institution{Takoma Park Middle School}
   \city{Takoma Park}
   \country{USA}
 \email{ihajiaghayi@gmail.com}
 }
 \and
 {MohammadTaghi Hajiaghayi
 \affiliation{
   \institution{University of Maryland}
   \city{College Park}
   \country{USA}
 \email{hajiagha@umd.edu}
 }}
 \and
 {Gary Peng
 \affiliation{
   \institution{University of Maryland}
   \city{College Park}
   \country{USA}
 \email{gpeng1@terpmail.umd.edu }
 }}
 \and
 {Suho Shin
 \affiliation{
   \institution{University of Maryland}
   \city{College Park}
   \country{USA}
 \email{suhoshin@umd.edu}
 }}
}

\begin{document}
\date{}
\maketitle
% Abstract. Note that this must come before \maketitle.
\allowdisplaybreaks
\begin{abstract}
    We study bilateral trade with a broker, where a buyer and seller interact exclusively through the broker. The broker strategically maximizes her payoff through arbitrage by trading with the buyer and seller at different prices.
    We study whether the presence of the broker interferes with the mechanism's gains-from-trade (GFT) achieving a constant-factor approximation to the first-best gains-from-trade (FB), in a similar vein to the constant-factor approximability without a broker by Deng, Mao, Sivan and Wang (STOC'21).

    We first identify a structural connection between GFT, FB, and the broker's expected profit when the broker uses a posted-pricing mechanism, which implies the $1/2$-approximation via median-based pricing by McAfee (AERB'08) as a special case.
    More importantly, inspired by this, we show that the GFT achieves a $1 / 36$-approximation to the FB even if the broker runs an optimal posted-pricing mechanism under symmetric agents with monotone-hazard-rate distributions, which are central to mechanism design due to Myerson (MOR'81).
    We cement this result by proving that the approximation factor can also be lower bounded by $2/M^2$ if the hazard rates are bounded above by $M$.
    Furthermore, if the broker uses a quantile-based posted-pricing mechanism of offering the $\alpha$-quantile (of the buyer's distribution) to the buyer and the $\beta$-quantile (of the seller's distribution) to the seller, \eg due to a lack of a precise estimation of the distributions, the mechanism achieves a $\beta(1-\alpha)$-approximation to the first-best GFT. \gpedit{As a corollary, we prove that there exists a simple single-sample mechanism of the broker that achieves a $1 / 12$-approximation to the first-best GFT.}
    
    Beyond posted-pricing mechanisms, even if the broker uses an arbitrary incentive-compatible (IC) and individually-rational (IR) mechanism that maximizes her expected profit, we prove that it induces a $1 / 2$-approximation to the first-best GFT when the buyer's and seller's distributions are uniform distributions with arbitrary supports. 
    This bound is shown to be tight.

    We complement such results by proving that if the broker uses an arbitrary profit-maximizing IC and IR mechanism, there exists a family of problem instances under which the approximation factor to the first-best GFT becomes arbitrarily close to zero.
    We show that this phenomenon persists even if we restrict one of the buyer's or seller's distributions to have a singleton support, or even in the symmetric setting where the buyer and seller have identical distributions. \gpedit{Finally, we show that in the general-distribution and public-seller settings, the inapproximability of the first-best gains-from-trade carries over to first-best social welfare as well.}
    
    \extrafootertext{Correspondence to Suho Shin. A part of this work has appeared in SODA'25. This work is partially supported by DARPA QuICC, NSF AF:Small \#2218678, NSF AF:Small \#2114269, Army-Research Laboratory (ARL) \#W911NF2410052, and MURI on Algorithms, Learning and Game Theory.}

\end{abstract}

\newpage
\section{Introduction}
The problem of bilateral trade, introduced by~\cite{myerson1983efficient}, has been a cornerstone in mechanism design and algorithmic game theory in the past few decades.
In bilateral trade, a single buyer and seller engage in trade, with each party holding private information about their valuation.
If the buyer's valuation is larger than the seller's, social welfare strictly increases by transferring the item from the seller to buyer at certain prices.
The buyer and seller, however, may strategically hide their true valuation to increase their own payoff.
This leads to the challenge of constructing a truthful mechanism while maintaining its efficiency in terms of the social welfare.
More precisely, the objective here is to design a mechanism that optimizes efficiency while satisfying (i) incentive-compatibility (IC), \ie both players truthfully report their values, (ii) individual rationality, \ie the players are not worse off by participating the trade, and (iii) budget-balancedness (BB) so that the mechanism does not run a deficit.

The seminal result by~\cite{myerson1983efficient}, however, reveals that it is impossible to achieve ex-post efficiency\footnote{A mechanism is ex-post efficient if it achieves the optimal ex-post social welfare, \ie if trade occurs whenever the buyer's valuation is at least the seller's valuation.} subject to the constraints of IC, IR, and BB, and even subject to the loosened notion of Bayes-Nash IC (BNIC), IR, and weakly budget-balancedness (WBB).
Correspondingly, a long line of work~\citep{mcafee2008gains,blumrosen2014reallocation,blumrosen2016almost,blumrosen2016approximating} has studied the best possible approximate efficiency, in particular for the notion of \emph{gains-from-trade} (GFT),\footnote{Gains-from-trade is defined by the expected marginal increase of the social welfare. Thus, it is typically harder to approximate than social welfare. We refer to Section~\ref{sec:model} for more details.} with respect to those desideratum.
It was a long-standing open problem whether a constant-factor approximation to the first-best GFT was possible, until it was recently answered in the affirmative by~\cite{deng2022approximately}.

On the other hand, in many real-world scenarios, the buyer and seller might not be able to trade the item directly, \eg due to regulatory requirements, privacy concerns, or time constraints, in particular for financial markets, real estate, over-the-counter markets, and markets for demerit goods with licensed intermediaries.
Brokers do more than just facilitate transactions, where they often play a crucial role in price discovery and risk management.
Notably, \cite{myerson1983efficient} also study bilateral trade with a broker and characterizes an optimal mechanism for a broker who tries to maximize her own profit.
% In these applications, the buyer and seller typically can only interact exclusively through a broker who governs the entire trade, \eg licensed intermediaries for demerit goods.
% Beyond mere intermediation, brokers play a significant role in price discovery, gauging demand
% and supply to determine optimal asset prices
% This setting too has been studied by~\cite{myerson1983efficient}.
% The paper 
Despite being significantly practical and well-founded, there have only been a few papers studying bilateral trade in the presence of broker~\citep{vcopivc2008robust,zhang2019efficient,eilat2021bilateral,kuang2023profit}.

In this work, we initiate the study of approximating gains-from-trade in bilateral trade with a broker, investigating the following fundamental question, which has not been explored before:
\definecolor{mycolor}{rgb}{0.9, 0.90, 0.9}
\begin{tcolorbox}[colback=mycolor,colframe=gray!75!black,colframe=white]
\begin{center}
        \emph{Does the presence of a strategic broker interfere with the possibility of constant-factor approximation to the first-best gains-from-trade?}
\end{center}
\end{tcolorbox}
We answer this question both affirmatively and negatively by characterizing the regimes under which constant approximation is possible or not.
Our characterization reveals the extent to which the marginal increase of the social welfare can be interrupted by a strategic broker, shedding light on the public decision-making process of enforcing regulatory restriction in trade. 

% \sscomment{Say sth about the implication of our results, interesting points...}
% Our results imply 

\subsection{Our Results}\label{sec:summary-results}
To summarize our results, we first briefly introduce the problem setup.
In bilateral trade with a broker, there exist a buyer and a seller with private valuations $v$ and $c,$ respectively. We use $F$ and $G$ to denote the respective cumulative probability distributions (CDFs) of the buyer and the seller, and $f$ and $g$ to denote their respective probability density functions (PDFs). In addition, we write $\mu_F(\alpha)$ for $\alpha \in [0,1]$ to denote the $\alpha$-quantile of $F,$ \ie $\mu_F(\alpha) = \inf\{t: F(t) \ge \alpha\}$, and similarly define $\mu_G(\beta)$  for $\beta \in [0,1]$ to be the $\beta$-quantile of $G$. 
The buyer and the seller exclusively trade through a broker, who tries to maximize her expected profit by constituting an arbitrage on the prices to the buyer and seller. We consider \emph{Bayesian mechanisms} in which the broker knows $F$ and $G$.\footnote{In comparison, the literature often considers \emph{prior-independent} mechanisms that do not know the distributions, \eg second-price auctions.} The gains-from-trade in this setting is defined to be the expected marginal increase of the social welfare, \ie by $\GFT = \Exu{v \sim F, c \sim G}{(v-c) \cdot x(v,c)},$ where $x(v,c)$ denotes the probability that the item is transferred from the seller to the buyer if their respective valuations are $c$ and $v$.
If there exists an omniscient planner who can observe the private valuations and also enforce trade at arbitrary prices, society would enjoy the maximum increase of social welfare, \ie the first-best $\GFT$ (henceforth $\FB$), defined as $\FB = \Exu{v \sim F, c \sim G}{(v-c) \cdot \Ind{v \ge c}}$. We refer to Section~\ref{sec:model} for more detailed definitions and notation.

\paragraph{Distribution Types}
We say that the buyer's distribution is \emph{regular} if his virtual valuation $\phi_F(x) = x - (1-F(x))/f(x)$ is nondecreasing, and we say that the seller's distribution is regular if his virtual valuation $\phi_G(x) = x + G(x)/g(x)$ is nondecreasing.
The quantities $h_b(x) = f(x)/(1-F(x))$ and $h_s(x) =g(x)/G(x)$ are usually denoted as \emph{hazard rates}.
We say that $F$ has a \emph{monotone hazard rate} (MHR) if $h_b(x)$ is nondecreasing, and we say that $G$ is MHR if $h_s(x)$ is nonincreasing.
MHR and regular distributions are widely studied in the literature, with their fruitful connections to mechanism design (\cite{myerson1981optimal,blumrosen2016almost,blumrosen2016approximating}). Finally, we often consider the ex-ante \emph{symmetric} agents setting, in which the buyer's and seller's distributions are the same, \ie $F = G$.\footnote{Such ex-ante symmetric agents are well-motivated by size-discovery mechanisms in financial markets -- in particular for dark pools~\citep{kang2019fixed} -- and are widely assumed in the literature, \eg see~\cite{degryse2009dynamic,zhu2014dark,duffie2017size}.}

\paragraph{Posted-Pricing Mechanisms}
We begin by investigating a broker who runs a \emph{posted-pricing} mechanism,\footnote{This is often called a fixed-price mechanism in the literature.} \ie a broker who offers a pair of fixed prices $p$ and $q$ to the buyer and seller, respectively. Posted-pricing mechanisms are widely adopted in many real world applications due to their simplicity and transparency, in particular for online labor market and size-discovery mechanisms, including workup mechanisms and dark pools (\cite{kang2019fixed}).
In addition, in the standard setting of bilateral trade without a broker, posted-pricing mechanisms were shown by~\cite{hagerty1987robust} to be the only class of mechanisms satisfying budget-balancedness and DSIC.
Note further that posted-pricing mechanisms are clearly IR.

Towards obtaining approximation factors to FB for various regimes of problem instances, we first identify a structural inequality involving $\GFT$, the broker's expected profit, and $\FB$, which might be of independent interest.
To this end, we use a variant of the analysis by~\cite{mcafee2008gains} after decomposing the gains-from-trade using the posted prices.
% \sscomment{Make it proposition}
\begin{theorem}[Decomposition]\label{thm:pp-structure}
    Given a broker's posted-pricing mechanism with prices $(p,q)$, the following holds:
    \begin{align}
        \GFT \ge \Pro + \min(G(q), 1-F(p))\parans{\FB - \int_{q}^{p}G(x)(1-F(x)) \, dx}.\label{eq:07051845}
    \end{align}
\end{theorem}
% This implies that if $F(p),F(q),G(p)$ and $G(q)$ are all bounded above and below by some absolute constants, then $\GFT$ remains as constant-factor approximate to the $\FB$, even if the gap $p-q$ is significantly small.
As a byproduct, we reobtain the structural result by~\cite{mcafee2008gains}, including the $1 / 2$-approximation to the first-best GFT via median-based pricing. We state this formally in the following corollary:
\begin{corollary}\label{cor:pp-half}
    Consider bilateral trade without a broker\footnote{Alternatively, a broker who tries to maximize the gains-from-trade instead of her expected profit.} and a posted-pricing mechanism with fixed price $p$.
    Then, we have
    \begin{align*}
        \GFT \ge \min(G(p),1-F(p)) \cdot \FB.
    \end{align*}
    As a special case, if the median of $F$ is at least that of $G$, then setting $p$ to be any value between these medians yield a $1 / 2$-approximation to the first-best GFT.
\end{corollary}
Note that Corollary~\ref{cor:pp-half} immediately follows from plugging $p = q$ into Theorem~\ref{thm:pp-structure}, which results in $\Pro = 0$ and $\displaystyle \int_{q}^p G(x)(1-F(x)) \, dx = 0$, 

% On the other hand, it is not immediately clear what approximation factors Theorem~\ref{thm:pp-structure} give us when a broker is in fact present.
% It is necessary to find a 

On the other hand, it is not immediately clear what approximation factors we can obtain from Theorem~\ref{thm:pp-structure} when a profit-maximizing broker is in fact present.
Let $(p^*, q^*)$ be the broker's optimal pair of prices.
As an initial attempt, suppose that the broker uses the $1/4$-quantile of the seller's distribution and the $3/4$-quantile of the buyer's distribution as prices for the seller and buyer, respectively, assuming that $\mu_F(3/4) \ge \mu_G(1/4)$.
For simplicity, let $p = \mu_F(3/4)$, $q = \mu_G(1/4)$, $\Delta = p-q$ and $\displaystyle H(p,q) = \int_q^p G(x)(1-F(x)) \, dx$.
Then, to lower bound the right-hand side of Theorem~\ref{thm:pp-structure}, we might naively upper bound the integral by $\Delta \cdot \max(G(p), 1-F(q))$.
However, this involves the gap $\Delta$ between the posted prices as well as the probabilities $G(p)$ and $1 - F(q)$, none of which are easy to deal with.

% Then, from Theorem~\ref{thm:pp-structure}, we have that
% \begin{align*}
%     \GFT \ge \Pro + \frac{1}{4}\parans{\FB - \int_{\mu_G(1 / 4)}^{\mu_F(3 / 4)} G(x)(1-F(x)) \, dx}.
% \end{align*}
% Thus, it is essential to properly upper-bound $\int_{\mu_G(1/4)}^{\mu_F(3/4)} G(x)(1-F(x)) \, dx$ to obtain a reasonable lower bound on the RHS.
% One might naively upper bound this by
% \begin{align*}
%     \int_{\mu_G(1 / 4)}^{\mu_F(3 / 4)} G(x) (1-F(x)) \, dx 
%     &\le \left(\mu_F\left(\frac{3}{4}\right)-\mu_G\left(\frac{1}{4}\right)\right) \max\left(G\left(\mu_F\left(\frac{3}{4}\right)\right),1-F\left(\mu_G\left(\frac{1}{4}\right)\right)\right).
%     % \\
%     % &\le
%     % \left(\mu_F(\frac{3}{4})-\mu_G(\frac{1}{4})\right) \max(G(\mu_G(\frac{1}{4})),1-F(\mu_F(\frac{3}{4})))
%     % \\
%     % &\le 
%     % (\mu_F(\frac{3}{4})-\mu_G(\frac{1}{4})) \cdot \frac{1}{4}
% \end{align*}
% However, this relies on the gap between the prices $(\mu_F(3/4)-\mu_G(1/4))$ as well as the unknown probabilities $G(\mu_F(3/4))$ and $F(\mu_G(1/4)).$ 

% Instead,
% Now we proceed further to obtain whether one could recover constant factor approximation with an optimal posted pricing mechanism.
\ssedit{
Instead, taking a closer look at~\eqref{eq:07051845}, we consider two cases where $\displaystyle H(p,q) = \int_q^p G(x)(1-F(x))$ is sufficiently small or moderately large, respectively.
If $H(p,q)$ maintains a rather small value, then since
\begin{align*}
    \Pro = (p-q)(1-F(p))G(q) \le \int_q^p (1-F(x))G(x) \, dx,
\end{align*}
it should be the case that $\Pro$ is also small.
Thus, there is a high chance that $\min(G(q),1-F(p))$ is small as well, \ie we cannot find an appropriate lower bound of the right-hand side of~\eqref{eq:07051845} in multiplicative terms of $\FB$.

On the other hand, suppose $H(p,q)$ is moderately large, say $H(p,q) \ge C \cdot \FB$ for some constant $C > 0$.
% For ease of exposition, let us assume that $\alpha_F(3/4) \ge \alpha_G(1/4)$, and let $\Delta = \alpha_F(3/4) \ge \alpha_G(1/4)$.
Then, since the profit-maximizing broker's optimal payoff should be at least the payoff obtained from posting $(p,q)$, we know that
\begin{align*}
    \Pro(p^*, q^*) \ge \Pro(p,q) \ge \Delta \cdot \frac{1}{16},
\end{align*}
where we obtain the $1 / 16$ due to our choice of $(p,q)$ above.
At the same time, $\Delta$ provides a trivial upper bound on $H(p,q),$ which implies that we can lower bound the optimal broker's profit by $C' \cdot \FB$ for some constant $C' > 0$.
The main technical challenge here is to prove that there exists an appropriate choice of a pair of prices $(p, q)$ such that the corresponding quantile probabilities are absolute constants and $H(p,q)$ is at least some constant times $\FB$.
We show that this is indeed possible for MHR distributions with symmetric agents, using a clever decomposition and an exploitation of the properties of MHR distributions.
Formally, we prove the following theorem:
\begin{theorem}[Symmetric agents, MHR distributions]\label{thm:pp-optimal}
    Consider symmetric agents with MHR distributions. Then, if the broker runs an optimal posted-pricing mechanism,\footnote{Note that the result here only holds for an optimal posted-pricing mechanism, not an arbitrary posted-pricing mechanism. Indeed, if the broker is adversarial rather than profit-maximizing, one cannot obtain any approximation factor since she may give up her profit but try to ruin society by making no trade at all.} it follows that $\GFT \ge \nicefrac{1}{36} \cdot \FB$.\footnote{Note that Theorem~\ref{thm:pp-optimal} immediately implies that under symmetric agents with MHR distributions, any optimal posted-pricing mechanism achieves a $\nicefrac{1}{36}$-approximation to the first-best social welfare.} \sscomment{A footnote for the Social Welfare.}
\end{theorem}
}
% To compare this with the results in the literature without a broker, we note that constant approximation is possible for a fixed price mechanism with identical distributions without a broker by~\cite{kang2019fixed}.
% It remains a major open problem whether constant approximation is still possible in our setting for more general class of distributions beyond (doubly) MHR.

It is worth noting that symmetric agents are well-motivated by size-discovery mechanisms including workup mechanisms by the U.S. Treasury market and dark pools~\citep{duffie2017size,zhu2014dark,kang2019fixed}.
In particular, in stock markets with brokerage, each participant can act either as a buyer or a seller. Thus, the broker cannot distinguish between the buyer and seller and hence cannot estimate their respective distributions separately. MHR distributions are also central to the theory of mechanism design along with regular distributions, due to the celebrated result by~\cite{myerson1981optimal} and further seminal works in algorithmic game theory~\citep{hartline2009simple,chawla2010multi,chawla2007algorithmic,hartline2008optimal} (see~\cite{roughgarden2010algorithmic} for a comprehensive exposition).

To tackle the relatively small approximation guarantee presented above, we show that if the hazard rates of the MHR distributions are further bounded above by $M > 0$, then any optimal posted-pricing mechanism achieves at least a $2/M^2$-approximation.
The proof follows from expanding $\FB$ and $\GFT$ in a slightly different manner than in Theorem~\ref{thm:pp-structure}, largely by exploiting symmetry and properties of MHR distributions.
The formal statement can be presented as follows:
\begin{theorem}[Symmetric agents, MHR distributions with bounded hazard rates]\label{thm:pp-mhr-bounded}
    Consider symmetric agents with MHR distributions whose hazard rates are uniformly bounded above by $M$.
    Then, if the broker runs an optimal posted pricing mechanism, it follows that $\GFT \ge 2/M^2 \cdot \FB$.\footnote{Note that $M$ is guaranteed to be at least $2$ in such problem instances, implying that the best possible approximation factor here is $1 / 2$.}
\end{theorem}

So far, we have only considered optimal posted-pricing mechanisms, which requires the broker to have complete information, \ie a precise estimation, about the distributions $F$ and $G$.
In practice, however, it is typically difficult for the broker to exactly estimate the probability distributions of the buyer and the seller.
Instead, the broker may only be aware of several statistics about the distributions such as expectations, medians, or more generally quantiles.
Therefore, our next question is to characterize the approximation factor if the broker runs a \emph{quantile-based} posted-pricing mechanism where she offers $\mu_F(\alpha)$ to the buyer and $\mu_G(\beta)$ to the seller.
In this case, we observe that one can express the approximation factor to the first-best GFT as a function of $\alpha$ and $\beta$ without introducing the gap between the quantiles itself, formalized as follows:

% Nevertheless, we observe that one can express the approximation factor to the first-best GFT as a function of the probabilities corresponding to the quantiles without introducing the gap between the quantiles itself, formalized as follows:

% We show that if the broker uses quantiles of the distributions as the posted prices, then the approximation factor can be parameterized by the probabilities corresponding to the quantiles.
% For instance, if the broker uses median 
% Let $\mu_F(\alpha)$ denote the $\alpha$-quantile of a cumulative distribution function $F$, \ie $\mu_F(\alpha) = \inf\{t: F(t) \ge \alpha\}$, and similarly define $\mu_G(\beta)$ be the $\beta$-quantile of $G$.

\begin{theorem}[Quantile-based posted-pricing mechanism]\label{thm:pp-quantile}
    Suppose $\mu_G(\beta) \le \mu_F(\alpha),$ with $\alpha < 1$ and $\beta > 0.$ Then, if the broker sets $p = \mu_F(\alpha)$ and $q = \mu_G(\beta),$ we have that
    \begin{align*}
        \GFT \ge \beta(1-\alpha)\FB.
    \end{align*}
\end{theorem}

Alternatively speaking, since we know that $\Pro = (\mu_F(\alpha) - \mu_G(\beta))(1-\alpha)\beta$, this is equivalent to $\GFT \ge \Pro / \Delta \cdot \FB,$ where $\Delta$ denotes the gap between the prices. As an example of Theorem~\ref{thm:pp-quantile}, if the broker sets the $3/4$-quantile (\ie the third quartile) of the buyer's distribution as the price for the buyer and the $1/4$-quantile (\ie the first quartile) of the seller's distribution as the price for the seller, then the gains-from-trade is at least $1/16$ of the first-best GFT.

It is also worth noting that Theorem~\ref{thm:pp-quantile} immediately implies a weaker version of the result by~\cite{mcafee2008gains}.
Indeed, assuming that $\mu_G(1/2) \le \mu_F(1/2)$, plugging in $\alpha = \beta = 1/2$ yields a $1 / 4$-approximation to the $\FB$, which is strictly weaker than Corollary~\ref{cor:pp-half}.
Conceptually, therefore, Theorem~\ref{thm:pp-quantile} completely characterizes the approximation factor using the probabilities that correspond to the quantiles while sacrificing a reasonable amount of loss in the approximation factor.

\ssedit{
An interesting implication of this result is that under symmetric agents, a posted-pricing mechanism that draws a single sample from each of the (identical) distributions of the buyer and the seller and offers the larger sample as the price to the buyer and smaller sample as the price to the seller achieves a $\nicefrac{1}{12}$-approximation to $\FB.$ In other words, if the broker has limited information, then society can enjoy a constant fraction of the first-best gains-from-trade, even without any distributional assumptions such as MHR.
\begin{corollary}\label{cor:single-sample}
    Suppose $F = G.$ Then, given $p \sim F$ and $q \sim G,$ the posted-pricing mechanism $(\max(p, q), \min(p, q))$ obtains a $\nicefrac{1}{12}$-approximation to the first-best gains-from-trade in expectation.
\end{corollary}
}

\paragraph{(In)approximability beyond Posted-Pricing Mechanisms}
So far, we have restricted the broker to running a posted-pricing mechanism.
In the real world, however, the broker may run any arbitrary mechanism to maximize her profit.
Thus, we now consider the case where the broker runs any BNIC and IR mechanism that maximizes her profit.
The seminal result by~\cite{myerson1983efficient} characterizes one such optimal mechanism, which even turns out to be DSIC:
\begin{theorem}[\cite{myerson1983efficient}]\label{thm:myerson}
    Suppose the buyer's distribution $F$ and the seller's distribution $G$ are regular with bounded supports. Then, among all BNIC and IR mechanisms, the broker's expected profit is maximized by a mechanism that transfers the item from the seller to the buyer if and only if $\phi_F(v) \ge \phi_G(c)$.
\end{theorem}

Note that due to the strategic behavior of the broker, the above mechanism may result in strictly less trading than a mechanism that maximizes GFT. 
Indeed, if one considers the symmetric setting where both buyer's and seller's valuations have uniform distribution $U[0,1]$, one can observe that under the above mechanism, trade occurs if and only if $\phi_F(v) = 2v - 1 \ge \phi(c) = 2c$, \ie if and only if $v-c \ge 1/2$. However, as shown by~\cite{myerson1981optimal}, a BNIC and IR mechanism that maximizes the gains-from-trade in this example transfers the item if and only if $v - c \ge 1/4$. Thus, the existence of the broker interferes with social welfare. For uniform distributions, however, we show that a $1/2$-approximation to the first-best GFT is still possible, and that this bound is even tight.
\begin{theorem}[Uniform distribution]\label{thm:uniform}
    If $F$ and $G$ are uniform distributions, then $\GFT \ge \FB/2$ for any BNIC and IR mechanism that maximizes the broker's expected profit. Furthermore, this bound is tight in the sense that there exists a problem instance in which the approximation ratio is exactly $1 / 2$.
\end{theorem}
The main technical subtlety in proving Theorem~\ref{thm:uniform} lies in handling a number of cases depending the supports of the uniform distributions. To simplify the possible cases here, we introduce a result stating that $\GFT / \FB$ is invariant to stretches and shifts in the buyer's and seller's distributions as long as we apply the same transformations to both distributions, formally presented in Lemma~\ref{lem:normalization}.

On the other hand, if we consider general distributions, the extent to which the broker's strategic behavior degrades the GFT can be significant, formally stated as follows:
\begin{theorem}[Inapproximability]\label{thm:inapx-general}
    For any $\eps > 0$, there exists a problem instance such that $\GFT < \eps \cdot \FB$ for any BNIC and IR mechanism that maximizes the broker's expected profit.
\end{theorem}
The proof here is based on the characterization in Theorem~\ref{thm:myerson}. To briefly summarize, we start by trying to come up with a regular distribution for the buyer such that $\phi_F(v) < \phi_G(c)$ for all $v$ and $c,$ \ie such that trade never occurs under any BNIC and IR mechanism that maximizes the broker's expected profit. Solving a resulting differential equation, we obtain $F(x) = -a/x + 1$ for some constant $a > 0$.
This distribution, however, has an unbounded support of $[a,\infty)$, so we cannot guarantee the optimality of the mechanism given in Theorem~\ref{thm:myerson}, which only applies to bounded supports. 
Indeed, one can see that the above choice of $F$ results in $\Pro = 0,$ whereas a simple posted-pricing mechanism can achieve a strictly positive expected profit (assuming the problem instance is nontrivial, \ie $\FB > 0$). Thus, we need to truncate this distribution carefully to obtain the desired result.

One natural follow-up question to Theorem~\ref{thm:inapx-general} is to what extent such inapproximability persists.
We first consider a restricted set of problem instances where exactly one of the buyer and seller has public value.
Formally speaking, \emph{public-buyer} refers to the case where the buyer's ex-post valuation $v$ is deterministically known to the broker, and similarly for \emph{public-seller}.
We require analogues of Theorems 3 and 4 in~\cite{myerson1983efficient} to characterize an optimal BNIC and IR mechanism tailored to these specific problem instances, since we cannot directly apply~\cite{myerson1983efficient} when there are probability distributions with singleton support. These lemmas yield the following inapproximability result:
\begin{theorem}[Inapproximability with public seller or public buyer]\label{thm:inapx-public}
    For any $\eps > 0$, there exists (i) a continuous distribution $F$ and a singleton distribution $G$ and (ii) a singleton distribution $F$ and a continuous distribution $G$ such that $\GFT < \eps \cdot \FB$ for any BNIC and IR mechanism that maximizes the broker's expected profit.
\end{theorem}

Furthermore, we observe that inapproximability also carries over when the agents are symmetric.
\begin{theorem}\label{thm:inapx-symmetric}
    For any $\eps > 0$, there exist identical distributions $F = G$ such that $\GFT < \eps \cdot \FB$ for any BNIC and IR mechanism that maximizes the broker's expected profit.
    % $$$$ where $\GFT$ is the gains-from-trade of the mechanism and $\FB$ is the gains-from-trade of any ex-post-efficient mechanism.
\end{theorem}

\gpedit{
Finally, we show that in the general-distribution and public-seller settings, the inapproximability of the first-best gains-from-trade carries over to the first-best social welfare as well.
\begin{theorem}\label{thm:inapx-sw}
    For any $\epsilon > 0,$ there exists distributions $F$ and $G$ such that $\SW < \epsilon \FBW$ for any BNIC and IR mechanism that maximizes the broker's expected profit. Furthermore, this statement holds even if we restrict $G$ to be deterministic.
\end{theorem}
}

% \sscomment{Check $F(x) = 1 - exp(-(x*log(1 + x) - x + log(1 + x)) / A)$}
% \gpcomment{TODO: check this}
% \ssedit{
% It is also worth noting that the worst-case instance here has MHR distributions with bounded hazard rates, thus is strictly separated from the approximability result implied by Theorem~\ref{thm:pp-mhr-bounded}.
% That is, the gains-from-trade with a broker running optimal mechanism can be arbitrarily bad than who runs the optimal posted pricing mechanism.
% }

\subsection{Further Related Work}
\paragraph{Bilateral Trade}
The theory of bilateral trade dates back to the seminal result by~\cite{myerson1983efficient}, where the authors proved that it is impossible to have an ex-post efficient mechanism that satisfies BNIC, IR and WBB.
Since then, there has been a number of attempts to identify the best possible approximation ratio to the first-best gains-from-trade.
\cite{blumrosen2016almost} and \cite{blumrosen2016approximating} prove that a constant approximation ratio is possible for certain classes of distributions, \eg regular and MHR, and provide some worst-case upper bounds on the approximation ratios.
More recently, \cite{deng2022approximately} provide a mechanism that achieves a $1 / 8.23$-approximation to the first-best gains-from-trade. \cite{fei2022improved} significantly improves this to $1 / 3.15$ and further characterizes the exact approximation factor of the seller-pricing mechanism, given that the buyer's distribution satisfies MHR. Most recently, \cite{dobzinski2024bilateral} study the problem of approximating gains-from-trade and social-welfare for bilateral trade when the buyer's and seller's valuations can be correlated.
Finally, another line of literature in bilateral trade focuses on identifying the best-possible approximation ratio for the class of posted-pricing mechanisms~\citep{liu2023improved,kang2019fixed,kang2022fixed}.
It is worth noting that some of the above work focuses on approximating the social welfare defined by $\Exu{v \sim F, c \sim G}{c + (v-c)\cdot \Ind{v \ge c}}$, which is much easier to approximate than gains-from-trade, since social welfare attains positive value even if trade does not occur.

\paragraph{Bilateral Trade with a Broker}
In comparison to the significant developments in bilateral trade, bilateral trade with a broker has been less studied, despite its significance and even its introduction in the same seminal paper by~\cite{myerson1983efficient}. In terms of prior work on this problem, \cite{myerson1983efficient} characterize one BNIC and IR mechanism that maximizes the broker's expected profit.
\cite{zhang2019efficient} studies the impact of an imperfectly-informed broker who has partial information about the buyer's and seller's values, \eg the broker knows the the bounded interval in which the values are contained.
\cite{kuang2023profit} consider a variant of the standard bilateral trade setting in which a mediator commits to a mechanism that possibly includes multiple rounds of private communication and runs a posted-pricing mechanism after the communication ends.
Notably, their setting does not consider prior information on the buyer's and seller's distributions, \ie their setting considers prior-independent mechanisms.
Thus, certain information should be exchanged to induce an efficient trade, and the solution concept therein is a perfect Bayes equilibrium.
In a similar vein, \cite{vcopivc2008robust} study robust prior-independent mechanisms when the buyer's and seller's valuations are discounted over time and hence both agents are eager to have the trade occur as soon as possible. In this setting, the mediator keeps the reported valuations of the buyer and the seller privately while trade is incompatible. Then, after trade becomes compatible, the mediator discloses the agreement and trade occurs at the agreed price.
We emphasize that none of the literature has studied the approximation ratio for social welfare or gains-from-trade in the standard bilateral trade with broker model introduced by~\cite{myerson1983efficient}.

\paragraph{Two-sided Markets}
Several works consider an extension of bilateral trade that is typically referred to as two-sided markets or double auctions.
In this setting, there are multiple sellers and buyers, each of which is interested in selling his own items and buying some items of interest.
The seminal result by~\cite{mcafee1992dominant} presents a double-auction, trade-reduction mechanism that is IR, DSIC, BB, and achieves a high ex-post efficiency under a large market.
Since then, several works have studied prior-independent mechanisms for double-auction settings~\citep{deshmukh2002truthful,baliga2003market}. \cite{babaioff2020bulow} and~\cite{cai2024power} study the Bulow-Klemperer~\citep{bulow1994auctions} result for gains-from-trade in double-auction settings, analyzing the power of recruiting more buyers and sellers for prior-independent mechanisms, and
\cite{brustle2017approximating} and \cite{cai2021multi} study the approximate mechanism for bilateral trade in double auctions in both the single-dimensional and multi-dimensional settings.

\section{Problem Setup}\label{sec:model}
\paragraph{Bilateral Trade}
In the bilateral trade problem, a seller tries to sell an item to a buyer.
Both the buyer's and the seller's valuations are private and independently drawn from cumulative distributions (CDFs) $F$ and $G$, respectively, both of which are supported on $\R_{\ge 0}$.
Following~\cite{myerson1983efficient}, we write $v$ and $c$ to denote the buyer's and the seller's respective valuations.
We write $f$ and $g$ to denote the density functions of $F$ and $G$, respectively.
% We assume that $v$ and $c$ are independent.

If trade occurs between the buyer and the seller with valuations $v$ and $c,$ respectively, then society realizes an increment of welfare of amount $v-c$.
If the allocation function of a mechanism (probability of trade) is $x(\tilde{v}, \tilde{c})$ given the reported valuations $\tilde{v}$ of the buyer and $\tilde{c}$ of the seller, then the \emph{gains-from-trade}\footnote{Note that this differs from social welfare since social welfare is $c$ even if trade does not occur.} (GFT) is defined by the expected amount of increased utility from trade, \ie
\begin{align*}
    \GFT = \Exu{v \sim F, c \sim G}{(v-c)\cdot x(v,c)},
\end{align*}
assuming truthful reporting of the buyer and the seller.
If there exists an omniscient social planner who can observe the private valuations, then he can enforce trade whenever the buyer's valuation is at least the seller's valuation.
Such an ideal mechanism induces $x(v,c) = \Ind{v \ge c}$, and the resulting GFT is defined as the \emph{first-best gains-from-trade} (\FB), \ie 
\begin{align*}
    \FB =  \Exu{v \sim F, c \sim G}{(v-c) \Ind{v \ge c}}.
\end{align*}
\gpedit{Similarly, the \textit{social welfare} (\SW) is defined by the expected utility from trade, \ie
\begin{align*}
    \SW = \Exu{c \sim G}{c} + \GFT,
\end{align*}
and the \textit{first-best social welfare} (\FBW) is defined as the social welfare of the ideal mechanism described above, \ie
\begin{align*}
    \FBW = \Exu{c \sim G}{c} + \FB.
\end{align*}}

A mechanism is \emph{ex-post efficient} if it achieves the first-best GFT.
However, the seminal result by~\cite{myerson1983efficient} asserts that it is not possible to construct a mechanism that is simultaneously (i) Bayesian-Nash incentive-compatible (BNIC), (ii) ex-interim individually rational\footnote{A mechanism is ex-interim IR if all the participants are not worse off by participating the mechanism when they know their own valuations but not the others' valuations.} (IR), (iii) weakly budget balanced\footnote{A mechanism is WBB if the payment from the buyer is always at least the payment to the seller in an ex-post manner.} (WBB), and (iv) ex-post efficient, if the intersection of the supports of $F$ and $G$ is an interval with nonzero measure.
Thus, the literature is typically focused on identifying the best-possible approximation factor to the first-best GFT of any BNIC, IR, and WBB mechanism~\citep{deng2022approximately,fei2022improved}.

\paragraph{Bilateral Trade with a Broker}
In our setup, in addition to the buyer and the seller, there is a broker who governs the entire trade, through whom the buyer and seller exclusively interact.
The broker is not interested in maximizing the GFT but instead tries to maximize her own profit by buying from the seller at a lower price and selling to the buyer at a higher price.
Following~\cite{myerson1983efficient}, we assume that the broker can be a net source or sink of monetary transfer but that the broker cannot own the item herself.
Alternatively, the broker realizes the value $0$ for the item and thus does not want to hold on to the item.

% \sscomment{Notation for $p,q$ is a bit confusing with the parameters of prices in posted pricing. May adjust the notations in this para later...}
Essentially, the broker's \emph{direct} mechanism $\Mec = (x,p)$ is defined by an allocation function $x: \R_{\ge 0}^2 \to [0,1]$ that maps the reported valuations of the buyer and the seller to the probability that the item is transferred from the seller to buyer and a payment function $p: \R_{\ge 0}^2 \to \R_{\ge 0}^2$ that maps the reported valuations to the expected payment $p = (p_s,p_b),$ where $p_s$ is the expected payment from the broker to the seller and $p_b$ the expected payment from the buyer to the broker.\footnote{Note that we consider expected payments here, as the mechanism may randomize the payments.} Note here that, assuming truthful reporting from the buyer and the seller, the broker's expected net profit can be defined as 
\begin{align*}
    \Pro = \int_{a_2}^{b_2}\int_{a_1}^{b_1}\parans{p_b(v,c) - p_s(v,c)} \, dF(v) \, dG(c).
\end{align*}
Similarly, assuming truthful reporting, the buyer's expected utility if his true valuation is $v$ can be written as 
\begin{align*}
    u_b(v) = \int_{a_2}^{b_2} \parans{x(v,c) \cdot v - p_b(v,c)} \, dG(c),
\end{align*}
and the seller's expected utility if his true valuation is $c$ can be written as
\begin{align*}
    u_s(c) = \int_{a_1}^{b_1} \parans{p_s(v, c) - x(v, c) \cdot c} \, dF(v).
\end{align*}

\paragraph{Desiderata}
We are interested in designing an \emph{ex-interim individually-rational} mechanism so that each agent's ex-interim payoff after observing his own valuation without knowing the other's is always nonnegative, \ie such that
\begin{align*}
    u_b(v) \ge 0, u_s(c) \ge 0
\end{align*}
for any $v$ and $c$.

A broker's mechanism is \emph{Bayesian-Nash incentive-compatible} (BNIC) if it satisfies 
\begin{align*}
    \Exu{c \sim G}{x(v,c) \cdot v - p_b(v,c)} &\ge \Exu{c \sim G}{x(\tilde{v},c) \cdot v - p_b(\tilde{v},c)}
\end{align*}
for all $v$ and $\tilde{v}$ and
\begin{align*}
    \Exu{v \sim F}{p_s(v, c) - x(v, c) \cdot c} &\ge \Exu{v \sim F}{p_s(v, \tilde{c}) - x(v, \tilde{c}) \cdot c}
\end{align*}
for all $c$ and $\tilde{c}.$ The first inequality corresponds to the buyer's ex-interim incentive to truthfully report his private value, and the second inequality corresponds to the seller's ex-interim incentive to truthfully report his private value. 
By the revelation principle~\citep{myerson1981optimal}, any BNIC mechanism can be simulated by a BNIC direct mechanism, so we consider BNIC direct mechanisms without loss of generality.

A broker's mechanism is \emph{dominant-strategy incentive-compatible} (DSIC) if it satisfies
\begin{align*}
    x(v,c) \cdot v - p_b(v,c) &\ge x(\tilde{v},c) \cdot v - p_b(\tilde{v},c)
\end{align*}
for every $v, \tilde{v},$ and $c$ and
\begin{align*}
    p_s(v, c) - x(v, c) \cdot c &\ge p_s(v, \tilde{c}) - x(v, \tilde{c}) \cdot c,
\end{align*}
for every $c, \tilde{c},$ and $v.$
This is a strictly stronger notion than BNIC, \ie  DSIC implies BNIC, since each player is always at least as well off by reporting his true valuation regardless of the other player's valuation.

In this paper, we seek to find the approximation factor to the first-best GFT for different classes of mechanisms.
Formally, a mechanism is an $\alpha$-approximation to the first-best GFT if it satisfies $\GFT \ge \alpha \FB.$ 

% Throughout, we often focus on several classes of distributions.
\paragraph{Distribution Types}
Finally, we introduce several classes of distributions that will be useful in presenting our results. We say that the buyer's distribution is regular if $\phi_F(x) = x - (1-F(x))/f(x)$ is nondecreasing, and we say that the seller's distribution is regular if $\phi_G(x) = x + G(x)/g(x)$ is nondecreasing.
% Note here that $\phi_F$ and $\phi_G$ are usually referred to as the virtual valuations of buyer and seller, respectively.
Note that this is a two-sided-market version of the standard regularity condition in single-sided auctions (\cite{myerson1981optimal}), where $\phi_F$ and $\phi_G$ are usually referred to as the virtual valuations of the buyer and seller, respectively.\footnote{In general, even if the density function does not exist, one can define the buyer's distribution to be regular if $\ln(1-F(x))$ is concave, and analogously for the seller.}
% We refer to~\cite{myerson1983efficient} for more details.
% If both the buyer's and seller's distributions are regular, we say that the problem instance is regular.
The hazard rates $h_b$ and $h_s$ of the buyer's and seller's distributions $F$ and $G$ are defined respectively as $h_b(x) = f(x)/(1-F(x))$ and $h_s(x) = g(x)/G(x)$.
We say that the buyer's (resp. seller's) distribution has monotone hazard rate (MHR) if $h_b(x)$ (resp. $h_s(x)$) is nondecreasing (resp. nonincreasing) for $x \ge 0$.
It is straightforward to check that MHR implies regularity, but not the other way around.
The family of MHR distributions includes the well-known Gaussian, exponential, and uniform distributions whose tails can be no heavier than an exponential distribution with hazard rate $0$.
However, regular distributions can be subsumed under a broader class of distributions with fatter tails, \eg $f(x) = x^{-(1+\alpha)}$ for some $\alpha \ge 1$.
Lastly, we often consider a symmetric setting in which the buyer's and seller's distributions are equivalent, \ie where $F = G$. \gpedit{Throughout our paper, we always assume that the expectations of $F$ and $G$ are finite.}

% For any nonnegative random variable $X$, it is \emph{not better than used in expectation} (NBUE) if it satisfies $\CEx{X-t}{X \ge t} \le \Ex{X}$ for all $t \ge 0$.

\section{Posted-Pricing Mechanisms}
We consider here a broker who implements a posted-pricing (fixed-price) mechanism, \ie she announces fixed prices to the buyer and seller, respectively.
Formally, a posted-pricing mechanism in our setting consists of two fixed prices $(p, q).$ The broker simultaneously offers to buy the item from the seller at price $q$ and sell the item to the buyer at price $p.$ Since the broker strategizes to maximize her profit, it follows that $p \ge q.$ After the broker makes her offers, the buyer and seller each decide whether to accept their respective offer in a \emph{take-it-or-leave} manner.
If both agree to their offers, then the trade occurs and the item is transferred from the seller to the buyer. Otherwise, no trade occurs.

Under a posted-pricing mechanism with prices $(p, q)$, the gains-from-trade and the broker's expected profit can be written as
\begin{align*}
    \GFT &= \int_{0}^{q}\int_{p}^{\infty} (v-c) \, dF(v) \, dG(c).
    \\
    \Pro &= \int_{0}^{q}\int_{p}^{\infty} (p - q) \, dF(v) \, dG(c).
\end{align*}

From the above definition, it is clear that any posted-pricing mechanism that maximizes the expected profit for the broker satisfies
\begin{align*}
    (p^*, q^*) &= \argmax_{(p, q) \in 
    \R_{\ge 0}^2} \int_{0}^{q}\int_{p}^{\infty} (p-q) \, dF(v) \, dG(c)\\
    &= \argmax_{(p, q) \in \R_{\ge 0}^2} (p - q)G(q)(1 - F(p)).
\end{align*}
We seek to analyze the GFT of any such optimal mechanism compared to the first-best GFT. 

We first provide the proof of Theorem~\ref{thm:pp-structure}.
\begin{proof}[Proof of Theorem~\ref{thm:pp-structure}]
% \begin{proof}[Proof of Theorem~\ref{thm:pp-structure}]
    Observe first that
    \begin{align*}
        \GFT &= \int_{p}^{\infty}\int_{0}^{q} (v-c) \, dG(c) \, dF(v)
        \\
        &=
        \int_{p}^{\infty}\int_{0}^{q} (v-p + p-q + q - c) \, dG(c) \, dF(v)
        \\
        &=
        (p-q)(1-F(p))G(q) + \int_{p}^{\infty}\int_{0}^{q} (v-p) \, dG(c) \, dF(v) + \int_{p}^{\infty}\int_{0}^{q} (q - c) \, dG(c) \, dF(v).
        \\
        &=
        \Pro + G(q)\int_{p}^\infty (v-p) \, dF(v) + (1-F(p))\int_{0}^q (q-c) \, dG(c).
    \end{align*}
    Integrating by parts, we have that
    \begin{align*}
        \int_{p}^\infty(v-p) \, dF(v) 
        &= 
        \bigg[-(v-p)(1-F(v))\bigg]_{p}^\infty + \int_{p}^\infty (1-F(v)) \, dv 
        \\
        &= \int_{p}^\infty (1-F(v)) \, dv,
    \end{align*}
    where we use $1-F(v)$ instead of simply $F(v)$ for the integration-by-parts to make the first term exactly zero.\footnote{\gpedit{Since the expectation of $F$ is finite, we have that $\displaystyle v(1 - F(v)) = v\int_v^\infty dF(x) \le \int_v^\infty x \, dF(x)$
    and $\displaystyle \lim_{v \rightarrow \infty} \parans{\int_v^\infty x \, dF(x)} = \Exu{v \sim F}{v} - \lim_{v \rightarrow \infty} \parans{\int_0^v x \, dF(x)} = 0,$
    so $\displaystyle\lim_{v \rightarrow \infty} v(1 - F(v)) = 0.$}
    }

    Similarly, for the seller's integral, we have that
    \begin{align*}
        \int_{0}^q(q-c) \, dG(c)
        &=
        \bigg[(q-c)G(c)\bigg]_{0}^q + \int_{0}^q G(c) \, dc 
        \\
        &= \int_{0}^q G(c) \, dc.
    \end{align*}
    Hence, we have that
    \begin{align*}
        \GFT &= \Pro + G(q)\int_{p}^\infty (1-F(x)) \, dx +(1-F(p))\int_{0}^qG(x) \, dx\\
        &\ge
        \Pro + \min(G(q),1-F(p)) \parans{\int_{p}^\infty (1-F(x)) \, dx + \int_{0}^q G(x) \, dx}
        \\
        &\ge
        \Pro + \min(G(q),1-F(p)) \parans{\int_{p}^\infty (1-F(x))G(x) \, dx + \int_{0}^q G(x)(1-F(x)) \, dx}
        \\
        &=
        \Pro + \min(G(q),1-F(p)) \parans{\int_{0}^\infty(1-F(x))G(x) \, dx - \int_{q}^p (1-F(x))G(x) \, dx}.
    \end{align*}
    Finally, to complete the proof, we use the following equality proven by~\cite{mcafee2008gains}, which we reiterate here for completeness:
    \begin{align}
        \FB
        &= \int_{0}^\infty\int_0^v (v - c) \, dG(c) \, dF(v)\nonumber
        \\
        &= \int_{0}^\infty\parans{\bigg[(v-c)G(c)\bigg]_0^v + \int_{0}^v G(c) \, dc} \, dF(v)\nonumber
        \\
        &= \int_{0}^\infty \int_{0}^v G(c) \, dc \, dF(v)\nonumber
        \\
        &=
        \bigg[-(1-F(v))\int_0^v G(c)dc \bigg]_0^\infty + \int_0^\infty (1-F(v))G(v) \, dv\nonumber
        \\
        &=
        \int_{0}^\infty(1-F(v))G(v) \, dv,\label{eq:07060145}
    \end{align}
    \gpedit{where the limit converges to zero since $\displaystyle(1 - F(v))\int_0^v G(c) \, dc \le v(1 - F(v))$ and $\displaystyle\lim_{v \rightarrow \infty} v(1 - F(v)) = 0$ (see footnote above).}
    This finishes the proof.
\end{proof}

% A direct corollary of this fact is the seminal result by~\cite{mcafee2008gains}, who provide $1/2$-approximation of the median-based price (either of the buyer or seller's distribution) without broker.
% Precisely, putting $p = q$, we immediately obtain the following corollary since $\Pro = 0$ and $\displaystyle \int_{q}^p (1-F(x))G(x)dx = 0$.
% \begin{corollary}\label{cor:pp-half}
%     Consider a bilateral trade without broker\footnote{Alternatively, a broker who tries to maximize the gains-from-trade instead of the profit.} and a posted pricing mechanism with fixed price $p$.
%     Then, we have
%     \begin{align*}
%         \GFT \ge \min(G(p),1-F(p)) \cdot \FB.
%     \end{align*}
%     As a special case, if median of $F$ is at least that of $G$, then setting $p$ to be any value between these medians yield $1/2$-approximation.
% \end{corollary}

\paragraph{Optimal Prices: Symmetric Agents, MHR Distributions}
We here provide the proof of Theorem~\ref{thm:pp-optimal}.
As we discussed in Section~\ref{sec:summary-results}, the proof mainly follows from proving the existence of quantiles $\mu_F(\alpha)$ and $\mu_F(\beta)$ such that $1-\alpha$ and $\beta$ are absolute constants regardless of the distributions.

% Recall that
% \begin{align*}
    % \GFT(p,q) \ge \Pro + \min(G(q),1-F(p)) \parans{\FB - \int_{q}^p (1-F(x))G(x) dx}.
% \end{align*}
\begin{proof}[Proof of Theorem~\ref{thm:pp-optimal}]
Recall that for a cumulative distribution function $F$ and a constant $0 \le \alpha \le 1,$ we use $\mu_F(\alpha)$ to denote the $\alpha$-quantile of $F,$ \ie $\mu_F(\alpha) = \inf\{t : F(t) \ge \alpha\}.$ Now, 
let $0 \le \beta < \alpha \le 1$ be constants to be determined later. Furthermore, suppose that
\begin{align*}
    \int_{\mu_F(\beta)}^{\mu_F(\alpha)} (1-F(x))F(x) \, dx \ge \frac{1}{C}\FB
\end{align*}
for some constant $C \ge 1$ to be determined later.
Now, note that
\begin{align*}
    \GFT(p^*, q^*) 
    \ge \Pro(p^*, q^*) 
    &\ge \Pro(\mu_F(\alpha), \mu_F(\beta)) 
    \\
    &\ge (\mu_F(\alpha) - \mu_F(\beta))\beta(1 - \alpha),
\end{align*}
where $(p^*, q^*)$ are an arbitrary pair of prices that maximizes the expected profit of the broker.
Therefore, we have that
\begin{align*}
     \GFT(p^*, q^*) \ge (\mu_F(\alpha)-\mu_F(\beta))\beta(1 - \alpha) \ge \beta(1 - \alpha)\int_{\mu_F(\beta)}^{\mu_F(\alpha)} (1-F(x))F(x) \, dx \ge \frac{\beta(1 - \alpha)}{C}\FB.
\end{align*}
Thus, it remains to find constants $\alpha, \beta, C$ such that $\displaystyle\int_{\mu_F(\beta)}^{\mu_F(\alpha)}  (1-F(x))F(x) dx \ge \frac{1}{C} \FB$. It suffices to find constants $\alpha, \beta, C$ such that
\begin{align}
    (C - 1)\int_{\mu_F(\beta)}^{\mu_F(\alpha)}  (1-F(x))F(x) \, dx \ge \int_0^{\mu_F(\beta)} (1-F(x))F(x) \, dx + \int_{\mu_F(\alpha)}^\infty (1-F(x))F(x) \, dx,\label{eq:07050550}
\end{align}
since by~\eqref{eq:07060145}, we have that
\begin{align*}
    \FB 
    &= \int_{0}^\infty (1-F(x))F(x) \, dx
    \\ 
    &=
    \int_{0}^{\mu_{F}(\beta)} (1-F(x))F(x) + \int_{\mu_F(\beta)}^{\mu_F(\alpha)} (1-F(x))F(x) + \int_{\mu_F(\alpha)}^\infty (1-F(x))F(x).
\end{align*}

Now, since we are considering the symmetric setting with MHR distributions, we know that $(1-F(x))/f(x)$ is nonincreasing and that $F(x)/f(x)$ is nondecreasing. Hence, for the integral on the left-hand side of~\eqref{eq:07050550}, we have that
\begin{align*}
    \int_{\mu_F(\beta)}^{\mu_F(\alpha)} (1-F(x))F(x) \, dx 
    &\ge
    \int_{\mu_F(\beta)}^{\mu_F(\alpha)}  \frac{1-F(x)}{f(x)} f(x) F(x) \, dx  \ge \frac{1- \alpha}{f(\mu_F(\alpha))}\int_{\mu_F(\beta)}^{\mu_F(\alpha)} F(x) \, dF(x)
    \\
    \int_{\mu_F(\beta)}^{\mu_F(\alpha)}  (1-F(x))F(x) \, dx
    &\ge
    \int_{\mu_F(\beta)}^{\mu_F(\alpha)}  (1-F(x))\frac{F(x)}{f(x)}f(x) \, dx \ge \frac{\beta}{f(\mu_F(\beta))}\int_{\mu_F(\beta)}^{\mu_F(\alpha)} (1-F(x)) \, dF(x).
\end{align*}
Therefore, we have that
\begin{align}
    \int_{\mu_F(\beta)}^{\mu_F(\alpha)} (1-F(x))F(x) \, dx
    &\ge
    \frac{1}{2}\parans{\frac{1-\alpha}{f(\mu_F(\alpha))}\int_{\mu_F(\beta)}^{\mu_F(\alpha)} F(x) \, dF(x) + \frac{\beta}{f(\mu_F(\beta))}\int_{\mu_F(\beta)}^{\mu_F(\alpha)} (1-F(x)) \, dF(x)}\nonumber
    \\
    &=
    \frac{1}{2}\parans{\frac{1-\alpha}{f(\mu_F(\alpha))}\cdot\frac{\alpha^2 - \beta^2}{2}
    +
    \frac{\beta}{f(\mu_F(\beta))}\left(\alpha-\beta - \frac{\alpha^2-\beta^2}{2}\right)
    }\\
    &\ge \frac{1}{2}\min\left(\frac{\alpha^2 - \beta^2}{2}, \alpha-\beta - \frac{\alpha^2-\beta^2}{2}\right)\left(\frac{1-\alpha}{f(\mu_F(\alpha))} + \frac{\beta}{f(\mu_F(\beta))}\right).
    \label{eq:07050554}
\end{align}
Now, consider the right-hand side in~\eqref{eq:07050550}.
The first term can be upper bounded by
\begin{align}
    \int_0^{\mu_F(\beta)} (1-F(x))F(x) \, dx  
    &=
    \int_0^{\mu_F(\beta)} \frac{F(x)}{f(x)}(1-F(x)) f(x) \, dx \nonumber
    \\
    &\le \frac{\beta}{f(\mu_F(\beta))}\int_0^{\mu_F(\beta)} (1-F(x)) \, dF(x) \nonumber
    \\
    &=
    \frac{\beta}{f(\mu_F(\alpha))}\parans{\beta - \frac{\beta^2}{2}}.\label{eq:07050555}
\end{align}
Similarly, the second term can be bounded above by
\begin{align}
    \int_{\mu(\alpha)}^\infty (1-F(x))F(x) \, dx
    &= 
    \int_{\mu(\alpha)}^\infty \frac{1-F(x)}{f(x)}F(x) f(x) \, dx\nonumber
    \\
    &\le
    \frac{1-\alpha}{f(\mu_F(\alpha))}\int_{\mu(\alpha)}^1 F(x) \, dF(x)\nonumber
    \\
    &=
    \frac{1-\alpha}{f(\mu_F(\alpha))}\parans{\frac{1 - \alpha^2}{2}}.\label{eq:07050556}
\end{align}
Thus, by~\eqref{eq:07050555} and~\eqref{eq:07050556}, the right-hand side in~\eqref{eq:07050550} can be upper-bounded by
\begin{align}
    \int_0^{\mu_F(\beta)} (1-F(x))F(x) \, dx + \int_{\mu(\alpha)}^\infty (1-F(x))F(x) \, dx &\le \max\left(\beta - \frac{\beta^2}{2}, \frac{1 - \alpha^2}{2}\right)\left(\frac{1 - \alpha}{f(\mu_F(\alpha))} + \frac{\beta}{f(\mu_F(\beta))}\right)\label{eq:07050314}.
\end{align}
Hence, from~\eqref{eq:07050550}, we see that it suffices by~\eqref{eq:07050554} and~\eqref{eq:07050314} to find constants $\alpha, \beta, C$ such that
\begin{align*}
    \frac{C - 1}{2}\min\left(\frac{\alpha^2 - \beta^2}{2}, \alpha - \beta - \frac{\alpha^2 - \beta^2}{2}\right) \ge \max\left(\beta - \frac{\beta^2}{2}, \frac{1 - \alpha^2}{2}\right).
\end{align*}
By (roughly) numerically optimizing the parameters $\alpha, \beta$ and $C$ here, taking $\alpha = 2 / 3, \beta = 1 / 3,$ and $C = 4,$ we see that 
\begin{align*}
    \GFT(p^*, q^*) \ge \frac{\beta(1 - \alpha)}{C}\FB = \frac{1}{36}\FB.
\end{align*}
\end{proof}

\paragraph{Optimal Prices: Bounding via Hazard Rates}
We provide the proof of Theorem~\ref{thm:pp-mhr-bounded} here.
The proof starts with the observation that $\FB$ can be written in terms of the integral of $|v-c|$.
Then, noting that the optimal broker's profit maximizes $(p-q)(1-F(p))G(q)$, we lower bound the profit by an averaging argument over a desirable region.
Combining this with the bounded hazard rates completes the proof.
For simplicity, we restrict the support of the distributions here to be $[0,1]$ without loss of generality, \eg see Remark 2 in~\cite{deng2022approximately}.

% \sscomment{repeated}
% Without any incentive or regulatory restriction, however, if the broker completely knows the distribution $F$ and $G$, her optimal strategy would be to maximize her own profit in~\eqref{eq:pp-profit}, instead of using quantile-based prices in Theorem~\ref{thm:pp-quantile}.
% Thus, the most natural question here is whether one could recover constant factor approximation to the $\FB$ with the broker who actually uses the optimal pair of prices.

\begin{proof}[Proof of Theorem~\ref{thm:pp-mhr-bounded}]
    Note first that
    \begin{align*}
        \int_0^1 \int_v^1 |v - c| \, dF(c) \, dF(v) &= \int_0^1 \int_0^c |v - c| \, dF(v) \, dF(c) \\
        &= \int_0^1 \int_0^v |c - v| \, dF(c) \, dF(v) \\
        &= \int_0^1 \int_0^v |v - c| \, dF(c) \, dF(v)\\
        &= \FB.
    \end{align*}
    Thus, we have that
    \begin{align*}
        \FB &= \frac{1}{2}\int_0^1 \int_0^1 |v - c| \, dF(c) \, dF(v).
    \end{align*}
    On the other hand, since the broker uses an optimal pair of prices $(p^*, q^*)$, we have that
    \begin{align*}
        \Pro = (p^* - q^*)(1-F(p^*))F(q^*) \ge (p-q)(1-F(p))F(q),
    \end{align*}
    for any $0 \le q \le p \le 1$.
    Integrating both sides over the region of $\R_{\ge 0}^2$ that satisfies $p \ge q$, we obtain
    \begin{align}
        \frac{\Pro}{2}
        &\ge \int_0^1 \int_{q}^1 (p-q)(1-F(p))F(q)\,dp\,dq \nonumber
        \\
        &= \int_0^1 \int_{q}^1 (p-q)\frac{1-F(p)}{f(p)}\frac{F(q)}{f(q)}f(p)f(q) \,dp\,dq.\nonumber
        \\
        &\ge
        \int_0^1 \int_{q}^1 (p-q)\frac{1}{M^2}f(p)f(q) \,dp\, dq 
        \\
        &=\frac{1}{M^2} \int_0^1 \int_{q}^1 (p-q)f(p)f(q) \,dp\,dq \nonumber
        \\
        &=
        \frac{1}{M^2}\int_0^1 \int_{q}^1 (p-q) \, dF(p) \, dF(q),\label{eq:07012356}
    \end{align}
    where in the last inequality, we use the fact that both hazard rates are uniformly bounded above by $M$ and thus their inverses are uniformly bounded below by $1/M$.

    Analogously, by expanding the integral from the seller's side first, we have
    \begin{align*}
        \frac{\Pro}{2}
        &\ge \int_{0}^1 \int_{0}^p (p-q)(1-F(p))F(q) \,dq \,dp
        \\
        &= \int_0^1 \int_{0}^p (p-q)\frac{1-F(p)}{f(p)}\frac{F(q)}{f(q)}f(p)f(q) \,dq\, dp
        \\
        &\ge
        \frac{1}{M^2}\int_0^1 \int_{0}^p (p-q) \, dF(q) \, dF(p)
        \\
        &=
        \frac{1}{M^2}\int_0^1 \int_{0}^q (q-p) \, dF(p) \, dF(q).
    \end{align*}
    Combining this with~\eqref{eq:07012356}, we obtain
    \begin{align*}
        \Pro 
        &\ge \frac{1}{M^2}\int_{0}^1\int_{0}^1 |p-q| \, dF(p)dF(q)
        \\
        &= \frac{2}{M^2} \FB,
    \end{align*}
    and the proof follows from the fact that $\GFT \ge \Pro$.
    % and applying the similar argument as above, we obtain
    % \begin{align*}
    %     \Pro 
    %     &\ge \int_0^1 \int_{q}^1 (p-q)(1-F(p))G(q)dpdq
    %     \\
    %     &= 
    %     \int_{0}^1\int_{0}^p (p-q) (1-F(p))G(q) dq dp \explain{change the order of integration}
    %     \\
    %     &=
    %     \int_{0}^1\int_{0}^q (q-p) (1-F(q))G(p) dp dq,\explain{renaming variables}
    % \end{align*}
    % which further implies that
    % \begin{align*}
    %     \Pro \ge \int_{0}^1 \int_{0}^1 |p-q|
    % \end{align*}
    
    % \begin{align*}
    %     \Pro 
    %     &\ge \int_{0}^1\int_{0}^1 (p-q)(1-F(p))F(q)dp dq
    %     \\
    %     &= \int_{0}^1\int_{0}^1 (p-q)\frac{1-F(p)}{f(p)}\frac{G(q)}{g(q)} \cdot f(p)g(q) dp dq
    %     \\
    %     &=
    %     \int_{0}^1\int_{0}^1 (p-q)\frac{1-F(p)}{f(p)}\frac{G(q)}{g(q)} dF(p) dG(q).
    % \end{align*}
    % Note, however, that due to the bounded hazard rate both the fractions are lower bounded by $1/M$.
    % Thus, we have
    % \begin{align*}
    %     \Pro \ge \int_{0}^1\int_{0}^1 (p-q) dF(p)dF(q)
    % \end{align*}
\end{proof}

It is worth noting that the approximation factor here cannot be better than $1/2$, since we must have that $M \ge 2$ according to the problem setup.
To see why, observe that
\begin{align*}
    \frac{1-F(x)}{f(x)} \ge \frac{1}{M}, \frac{F(x)}{f(x)} \ge \frac{1}{M}
\end{align*}
implies that $M \ge 2$ as follows:
\begin{align*}
    \frac{1}{M} \le \frac{F(x)}{f(x)} = 1 - \frac{1 - F(x)}{f(x)} \le 1-\frac{1}{M}.
\end{align*}
\paragraph{Quantile-based Prices}
As mentioned in Section~\ref{sec:summary-results}, the decomposition in Theorem~\ref{thm:pp-structure} cannot be directly applied to obtain a constant approximation factor for the first-best GFT in general, as it depends on an integral involving the buyer's and seller's CDFs in the interval between the two posted prices.
This brings us to the question of whether we can obtain an approximation factor purely characterized by the quantile probabilities at the offered prices, as suggested by Theorem~\ref{thm:pp-quantile}.

We formally prove Theorem~\ref{thm:pp-quantile} in what follows:

\begin{proof}[Proof of Theorem~\ref{thm:pp-quantile}]
    For any $0 \le c \le \mu_G(\beta) \le \mu_F(\alpha),$
\begin{align}
    \int_c^\infty (v - c) \, dF(v) &= \int_c^{\mu_F(\alpha)} (v - c) \, dF(v) + \int_{\mu_F(\alpha)}^\infty (v - c) \, dF(v) \nonumber\\
    &\le \frac{\alpha}{1 - \alpha}\int_{\mu_F(\alpha)}^\infty (v - c) \, dF(v) + \int_{\mu_F(\alpha)}^\infty (v - c) \, dF(v) \nonumber\\
    &= \frac{1}{1 - \alpha}\int_{\mu_F(\alpha)}^\infty (v - c) \, dF(v) \label{eq:07051208},
\end{align}
where we used the fact that $$\int_c^{\mu_F(\alpha)} (v - c) \, dF(v) \le \frac{\alpha}{1 - \alpha}\int_{\mu_F(\alpha)}^\infty (v - c) \, dF(v).$$ Observe that this holds because (1) the probability mass of the first integral is at most $\alpha$ while the probability mass of the second integral is exactly $1-\alpha$ and (2) the value of $v - c$ at any point in the probability mass of the second integral is always as least as much as the value of $v - c$ at any point in the probability mass of the first integral.\ In fact, reusing this idea, we have that
\begin{align}
    \int_{\mu_G(\beta)}^\infty \int_c^\infty (v - c) \, dF(v) \, dG(c) &\le \frac{1 - \beta}{\beta}\int_0^{\mu_G(\beta)} \int_c^\infty (v - c) \, dF(v) \, dG(c) \nonumber\\
    \int_0^\infty \int_c^\infty (v - c) \, dF(v) \, dG(c) &\le \frac{1}{\beta}\int_0^{\mu_G(\beta)} \int_c^\infty (v - c) \, dF(v) \, dG(c) \label{eq:07051209}.
\end{align}
Finally, combining~\eqref{eq:07051208} and~\eqref{eq:07051209}, we obtain $$\int_0^{\mu_G(\beta)} \int_{\mu_F(\alpha)}^\infty (v - c) \, dF(v) \, dG(c) \ge \beta(1 - \alpha)\int_0^\infty \int_c^\infty (v - c) \, dF(v) \, dG(c),$$ as desired.
\end{proof}
\ssedit{
% One interesting implication of Theorem~\ref{thm:pp-quantile} is that a there exists a simple single-sample mechanism that admits a constant-factor approximation to the first-best gains-from-trade in the symmetric case.

\begin{proof}[Proof of Corollary~\ref{cor:single-sample}]
    Note that drawing samples $p \sim F$ and $q \sim G$ is equivalent to sampling quantiles $\alpha$ and $\beta$ of $F$ and $G,$ respectively, uniformly at random over $[0,1]$. Thus, letting $\GFT(x,y)$ denote the gains-from-trade of the posted-pricing mechanism $(x,y),$ we have that
    \begin{align*}
        &\Exu{p \sim F, q \sim G}{\GFT(\max(p, q), \min(p, q))}\\
        &\quad= \int_0^1 \int_{\beta}^1 \GFT(\mu_F(\alpha), \mu_G(\beta)) \, d\alpha \, d\beta + \int_0^1 \int_0^{\beta} \GFT(\mu_G(\beta), \mu_F(\alpha)) \, d\alpha \, d\beta\\
        & \quad \ge \int_0^1 \int_{\beta}^1 \beta(1 - \alpha)\FB \, d\alpha \, d\beta + \int_0^1 \int_0^{\beta} \alpha(1 - \beta)\FB \, d\alpha \, d\beta = \frac{1}{12}\FB,
    \end{align*}
    where the inequality holds by Theorem~\ref{thm:pp-quantile}, since $F = G.$
\end{proof}
We conclude by remarking that having access to more samples may not necessarily increase the approximation factor to the first-best gains-from-trade, since it may be more optimal for the better-informed broker to set prices that further restrict trade.
}

\section{Optimal Mechanisms}
We now consider the case where the broker runs any general BNIC and IR mechanism that maximizes her profit.
We assume here that the supports of $F$ and $G$ are $[a_1,b_1]$ and $[a_2,b_2]$, respectively, \ie both are bounded, following~\cite{myerson1983efficient}.\footnote{Indeed, if we consider $F$ and $G$ with unbounded supports, we observe that the characterization of~\cite{myerson1983efficient} of an optimal BNIC and IR mechanism does not carry over, which is elaborated on more in the proof of Theorem~\ref{thm:inapx-general}.}
By Theorem~\ref{thm:myerson}, the broker's optimal mechanism exhibits the following GFT, assuming regular distributions:
\begin{align*}
    \GFT_{F, G} &= \int_{a_2}^{b_2} \int_{a_1}^{b_1} (v - c)\mathbf{1}\left\{\phi_F(v) \ge \phi_G(c)\right\} \, dF(v) \, dG(c)\\  
    &= \int_{a_2}^{b_2} \int_{a_1}^{b_1} (v - c)\mathbf{1}\left\{v - c \ge \frac{1 - F(v)}{F'(v)} + \frac{G(c)}{G'(c)}\right\} \, dF(v) \, dG(c).
\end{align*}
As usual, the first-best GFT is given by 
\begin{align*}
    \FB_{F, G} =\int_{a_2}^{b_2} \int_{a_1}^{b_1} (v - c)\mathbf{1}\left\{v - c \ge 0\right\} \, dF(v) \, dG(c).
\end{align*}
% The seminal result by~\cite{myerson1983efficient} characterizes one such optimal mechanism:
% \begin{theorem}[\cite{myerson1983efficient}]
%     Suppose the buyer's and the seller's virtual valuations $\phi_F$ and $\phi_G$ are increasing on $[a_1,b_1]$ and $[a_2,b_2]$, respectively.
%     Then, among all incentive-compatible, individually-rational mechanisms, the broker's expected profit is maximized by a mechanism that transfers the item from the seller to the buyer if and only if $\phi_F(v) \ge \phi_G(c)$.
% \end{theorem}

% Note that due to the strategic behavior of the broker, the above mechanism results in less trading than any mechanism that maximizes GFT. 
% Indeed, if one considers the symmetric setting where both buyer's and seller's valuations have uniform distribution $U[0,1]$, one can observe that under the above mechanism, trade occurs if and only if $\phi_F(v) = 2v - 1 \ge \phi(c) = 2c$, \ie if and only if $v-c \ge 1/2$. However, as shown by~\cite{myerson1981optimal}, the BNIC and IR mechanism that maximizes the gains-from-trade in this example transfers the item if and only if $v - c \ge 1/4$. Thus, the existence of the broker interferes with social welfare.

\paragraph{Uniform distributions}
To prove Theorem~\ref{thm:uniform}, we require a technical lemma that will simplify our analysis. The lemma states that the factor to which the broker's optimal mechanism approximates the first-best GFT is invariant to stretches and shifts in $F$ and $G,$ provided that we perform the same stretches and shifts to $F$ and $G.$ 
\begin{lemma}\label{lem:normalization}
    Let $F$ and $G$ be regular distributions supported on $[a_1, b_1]$ and $[a_2, b_2],$ respectively. Define $F^*(v) = F((v - k_1) / k_2)$ and $G^*(c) = G((c - k_1) / k_2)$ for $k_1 \in \mathbb{R}$ and $k_2 \in \mathbb{R}_{> 0}.$ If $\GFT_{F, G} = \alpha\FB_{F, G},$ then $\GFT_{F^*, G^*} = \alpha\FB_{F^*, G^*}.$
\end{lemma}
\begin{proof}
We first prove that $F^*$ and $G^*$ are regular. Note that to show that $\phi_{F^*}(v^*)$ is nondecreasing on its support $[k_2a_1 + k_1, k_2b_1 + k_1],$ it suffices to show that $\phi_{F^*}(k_2v + k_1)$ is nondecreasing on $[a_1, b_1].$ Indeed, since
\begin{align*}
    \phi_{F^*}(k_2v + k_1) &= (k_2v + k_1) - \frac{1 - F^*(k_2v + k_1)}{(F^*)'(k_2v + k_1)}\\
    &= (k_2v + k_1) - \frac{1 - F(v)}{F'(v) / k_2}\\
    &= k_1 + k_2\phi_F(v)
\end{align*}
and $\phi_F(v)$ is nondecreasing, $F^*$ is regular. A similar argument will show that $G^*$ is regular.

Now, we have that
\begin{align*}
    \GFT_{F, G} &= \int_{a_2}^{b_2} \int_{a_1}^{b_1} (v - c)\mathbf{1}\left\{v - c \ge \frac{1 - F(v)}{F'(v)} + \frac{G(c)}{G'(c)}\right\} F'(v) \, dv \; G'(c) \, dc\\
    &= \int_{a^*_2}^{b^*_2} \int_{a^*_1}^{b^*_1} \frac{v^* - c^*}{k_2}\mathbf{1}\left\{\frac{v^* - c^*}{k_2} \ge \frac{1 - F(\frac{v^* - k_1}{k_2})}{F'(\frac{v^* - k_1}{k_2})} + \frac{G(\frac{c^* - k_1}{k_2})}{G'(\frac{c^* - k_1}{k_2})}\right\}\\
    &\quad\cdot \frac{1}{k_2}F'\left(\frac{v^* - k_1}{k_2}\right) \, dv^* \; \frac{1}{k_2} \; G'\left(\frac{c^* - k_1}{k_2}\right) \, dc^*,
\end{align*}
where we use the change of variables with $v^* = k_2v + k_1$ and $c^* = k_2c + k_1$.
This can further be expanded as
\begin{align*}
    \GFT_{F,G} 
    &= \int_{a^*_2}^{b^*_2} \int_{a^*_1}^{b^*_1} \frac{v^* - c^*}{k_2}\mathbf{1}\left\{\frac{v^* - c^*}{k_2} \ge \frac{1 - F^*(v^*)}{k_2(F^*)'(v^*)} + \frac{G^*(c^*)}{k_2(G^*)'(c^*)}\right\} (F^*)'(v^*) \, dv^* \; (G^*)'(c^*) \, dc^*\\
    &= \int_{a^*_2}^{b^*_2} \int_{a^*_1}^{b^*_1} \frac{v^* - c^*}{k_2}\mathbf{1}\left\{v^* - c^* \ge \frac{1 - F^*(v^*)}{(F^*)'(v^*)} + \frac{G^*(c^*)}{(G^*)'(c^*)}\right\} (F^*)'(v^*) \, dv^* \; (G^*)'(c^*) \, dc^*\\
    &= \frac{1}{k_2}\GFT_{F^*, G^*},
\end{align*}
where the last equality follows from $F^*$ and $G^*$ being regular distributions.

A similar argument will show that $$\FB_{F, G} = \frac{1}{k_2}\FB_{F^*, G^*}.$$ Thus, if $\GFT_{F, G} = \alpha\FB_{F, G},$ then $$\GFT_{F^*, G^*} = k_2\GFT_{F, G} = \alpha k_2\FB_{F, G} = \alpha\FB_{F^*, G^*}.$$
\end{proof}

The proof of Theorem~\ref{thm:uniform} is deferred to Appendix~\ref{proof:uniform}, as it requires much case work that includes cumbersome algebraic manipulations.
% The following theorem states that the broker's optimal mechanism achieves at least a $\nicefrac{1}{2}$-approximation to the first-best GFT in the symmetric, uniform setting and that this bound is tight (See Appendix~\ref{proof:uniform} for the proof):

% \begin{theorem}\label{thm:uniform}
%     If $F$ and $G$ are uniform distributions, then the worst-case approximation ratio of the gains-from-trade for any incentive-compatible, individually-rational mechanism that maximizes the broker's expected profit is $\nicefrac{1}{2}.$ Furthermore, this bound is tight in the sense that there exists a problem instance in which the approximation ratio is exactly $\nicefrac{1}{2}$.
% \end{theorem}

\paragraph{Inapproximability}
We provide the proof of Theorem~\ref{thm:inapx-general} here, which states that the extent to which the broker's optimal strategic behavior degrades the GFT can be significant.
% \sscomment{redundant}
% Interestingly, we observe that 
% \begin{theorem}
%     For any $\eps > 0,$ there exists a problem instance in which any incentive-compatible, individually-rational mechanism that maximizes the broker's expected profit cannot achieve better than an $\eps$-approximation to the first-best gains-from-trade.
% \end{theorem}
\begin{proof}[Proof of Theorem~\ref{thm:inapx-general}]
    Recall that by Theorem~\ref{thm:myerson}, the gains-from-trade for any BNIC and IR mechanism that maximizes the broker's expected profit, assuming regular distributions, is given by 
    \begin{align*}
        \mathsf{GFT} &= \iint (v - c) \mathbf{1}\left\{v - c \ge \frac{1 - F(v)}{F'(v)} + \frac{G(c)}{G'(c)}\right\} \, dF(v) \, dG(c).
    \end{align*}
   Thus, to construct $F$ and $G$ for which such a mechanism achieves a $0$-approximation to the gains-from-trade, we want to construct $F$ such that 
    \begin{align*}
        v = \frac{1 - F(v)}{F'(v)},
    \end{align*}
    for all $v,$ in which case we have that
    \begin{align*}
        v - c < \frac{1 - F(v)}{F'(v)} + \frac{G(c)}{G'(c)}
    \end{align*}
    for all $v$ and $c.$ Solving the above differential equation, we obtain the general solution 
    \begin{align*}
        F(v) = -\frac{a}{v} + 1
    \end{align*}
    on $[a, \infty)$ for $a > 0.$ However, since the result about the optimal broker's mechanism by~\cite{myerson1983efficient}  only applies to distributions with bounded supports, we now modify $F$ by first truncating it and then extending it using a tangent line.\footnote{Indeed, with the original $F,$ trade never occurs by construction, so $\Pro = 0.$ This implies that the characterization by~\cite{myerson1983efficient} does not apply as the broker can simply use a posted-pricing mechanism to obtain a positive profit.}
    In particular, we define
    $$F(v) = 
    \begin{cases} 
        \displaystyle-\frac{a}{v} + 1 &\text{if $a \le v \le b$}\vspace{6pt}\\
        \displaystyle\left(-\frac{a}{b} + 1\right) + \frac{a}{b^2}(v - b) &\text{if $b < v \le 2b$}
    \end{cases}$$
    for $b > a \ge 1.$ Note that $F$ is regular, since
    \begin{align*}
        \phi_F(v) &= \begin{cases} 
        0 &\text{if $a \le v \le b$}\\
        2v - 2b &\text{if $b < v \le 2b$}
    \end{cases}
    \end{align*}
    is increasing. Finally, let $G$ be the uniform (and thus regular) distribution on $[0, 1].$
    
    We will now show that for any $\epsilon > 0,$ there exists $b > a$ such that any BNIC and IR mechanism that maximizes the broker's expected profit induces $\mathsf{GFT} < \epsilon\FB.$

    By the construction of $F,$ we have that $\displaystyle v - c < \frac{1 - F(v)}{F'(v)} + \frac{G(c)}{G'(c)}$ for all $v \in (a, b),$ so
    \begin{align*}
        \GFT &= \frac{a}{b^2}\int_b^{2b} \int_0^1 (v - c)\mathbf{1}\left\{v - c \ge \frac{1 - F(v)}{F'(v)} + \frac{G(c)}{G'(c)}\right\} \, dc \, dv\\
        &\le \frac{a}{b^2}\int_b^{2b} \int_0^1 (v - c) \, dc \, dv\\
        % &= \frac{a}{b^2}\int_b^{2b} \left[vc - \frac{c^2}{2}\right]_0^1 \, dv\\
        &= \frac{a}{b^2}\int_b^{2b} \left(v - \frac{1}{2}\right) \, dv\\
        % &= \frac{a}{b^2}\left[\frac{v^2}{2} - \frac{v}{2}\right]_b^{2b}\\
        &= \frac{a}{b^2}\left(\frac{3b^2}{2} - \frac{b}{2}\right)\\
        &= \frac{a(3b - 1)}{2b}.
    \end{align*}
    In addition, by the above derivation,
    \begin{align*}
        \FB &= a\int_a^b \int_0^1 (v - c) \, dc \, \frac{1}{v^2} \, dv + \frac{a}{b^2}\int_b^{2b} \int_0^{1} (v - c) \, dc \, dv\\
        &= a\int_a^b \left[vc - \frac{c^2}{2}\right]_0^1 \frac{1}{v^2} dv + \frac{a(3b - 1)}{2b}\\
        &= a\int_a^b \left(v - \frac{1}{2}\right)\frac{1}{v^2} \, dv + \frac{a(3b - 1)}{2b}\\
        % &= a\left[\ln v + \frac{1}{2v}\right]_a^b + \frac{a(3b - 1)}{2b}\\
        &= a\left(\ln\left(\frac{b}{a}\right) + \frac{1}{2b} - \frac{1}{2a}\right) + \frac{a(3b - 1)}{2b}.
    \end{align*}
    Hence,
    \begin{align*}
        \frac{\GFT}{\FB} &\le \frac{a(3b - 1)}{2b} \bigg / \left(a\left(\ln\left(\frac{b}{a}\right) + \frac{1}{2b} - \frac{1}{2a}\right) + \frac{a(3b - 1)}{2b}\right),
    \end{align*}
    which approaches $0$ as $b \rightarrow \infty,$ so $\displaystyle\frac{\GFT}{\FB}$ approaches $0$ as $b \rightarrow \infty.$
\end{proof}

\paragraph{Inapproximability with a Public Agent}
Here, we formally prove Theorem~\ref{thm:inapx-public}, first for the public-seller setting and then for the public-buyer setting.
To prove the theorem for the public-seller setting, we first introduce some notation:
% We start with public seller setting
% Interestingly, we observe that inapproximability holds even in a more restricted setting where one of the buyer's or seller's distributions has a singleton support, \ie is deterministic.
% Here, we consider the case where the buyer's distribution is an arbitrary continuous distribution supported on $[a, b]$ and the seller's distribution is a singleton distribution supported on $\{c\}.$ We first define some notation:
% \begin{definition}
%     A \textit{direct mechanism} $(p, x_B, x_S)$ is characterized by three functions, $p(\cdot), x_B(\cdot),$ and $x_S(\cdot).$ Given that the buyer reports $v$ as his valuation, $p(v)$ is the probability that a trade occurs, $x_B(v)$ is the expected payment from the buyer to the intermediary, and $x_S(v)$ is the expected payment from the intermediary to the seller. (Since the seller's distribution is singleton, there is no need for the seller to report her valuation.) 
% \end{definition}
\begin{definition}
    Given a direct mechanism $\Mec = (x(v), p(v))$\footnote{In the public-seller setting, the allocation function $x$ and payment function $p$ of $\Mec$ only depend on the buyer's reported valuation $v.$}, we define the following functions for the expected utilities of the buyer, intermediary, and seller, respectively, assuming truthful reporting by the buyer:
    \begin{itemize}
        \item $\displaystyle u_b(v) = x(v) \cdot v - p_b(v)$ 
        \item $\displaystyle u_i = \int_a^b (p_b(v) - p_s(v)) \, F'(v) \, dv$
        \item $\displaystyle u_s = \int_a^b (p_s(v) - x(v) \cdot c) \, F'(v) \, dv$
    \end{itemize}
\end{definition}
Note that with a public seller, $\Mec$ is BNIC if $u_b(v) \ge x(\tilde{v}) \cdot v - p_b(\tilde{v})$ for every $v$ and $\tilde{v},$ and $\Mec$ is IR if $u_b(v) \ge 0$ for all $v$ and $u_s \ge 0.$

% \begin{definition}
%     We say that a direct mechanism $(p, x_B, x_S)$ is \textit{incentive-compatible} if $U_B(v) \ge vp(v') - x_B(v')$ for all $v, v'.$ We say that a direct mechanism is \textit{individually-rational} if $U_B(v) \ge 0$ for all $v$ and $U_S \ge 0.$
% \end{definition}
% By the revelation principle, any incentive-compatible mechanism can be simulated by an incentive-compatible direct mechanism, so we consider direct mechanisms without loss of generality, and we will henceforth simply refer to direct mechanisms as mechanisms.

The following theorem is an analogue of Theorem 3 in~\cite{myerson1983efficient}:
\begin{theorem}
\label{singleton-seller-ic}
For any BNIC mechanism $\Mec$ in the public-seller setting, $u_b(v)$ is increasing and $$u_i + u_b(a) + u_s = \int_a^b \left(\left(v - \frac{1 - F(v)}{F'(v)}\right) - c\right) x(v) \, F'(v) \, dv.$$
\end{theorem}
\begin{proof}
Since $\Mec$ is BNIC, for every $v$ and $v',$ $$u_b(v) = x(v) \cdot v - p_b(v) \ge x(v') \cdot v - p_b(v')$$ and $$u_b(v') = x(v') \cdot v' - p_b(v') \ge x(v) \cdot v' - p_b(v).$$ Thus, $$(v - v')x(v) \ge u_b(v) - u_b(v') \ge (v - v')x(v').$$ In particular, if $v > v',$ then $x(v) \ge x(v'),$ so $x(v)$ is increasing and hence integrable. Therefore, $u'_b(v) = x(v)$ almost everywhere, so 
\begin{equation}
u_b(v) = u_b(a) + \int_a^v x(t) \, dt\label{eq:u_b_integral}
\end{equation}
and $u_b(\cdot)$ is increasing.

Now, by~\eqref{eq:u_b_integral}, we have that
\begin{align*}
    \int_a^b (v - c)x(v) \, F'(v) \, dv - u_i &= \int_a^b (vx(v) - p_b(v)) \, F'(v) \, dv + \int_a^b (p_s(v) - cx(v)) \, F'(v) \, dv\\
    &= \int_a^b u_b(v)F'(v) \, dv + u_s\\
    &= \int_a^b \left(u_b(a) + \int_a^v x(t) \, dt\right)F'(v) \, dv + u_s \\
    &= u_b(a) + \int_a^b \int_a^v x(t) \, dt \, F'(v) \, dv + u_s\\
    &= u_b(a) + \int_a^b \int_t^b  F'(v) \, dv \, x(t) \, dt + u_s\\
    &= u_b(a) + \int_a^b (1 - F(t)) \, x(t) \, dt + u_s.
\end{align*}
Finally, from the first and last expressions in this chain of equations, we obtain $$u_i + u_b(a) + u_s = \int_a^b \left(\left(v - \frac{1 - F(v)}{F'(v)}\right) - c\right) x(v) \, F'(v) \, dv.$$
\end{proof}

Next, the following is an analogue of Theorem 4 in~\cite{myerson1983efficient}, tailored to the public-seller setting.
\begin{theorem}\label{thm:publicseller-char}
    Suppose $F$ is a regular distribution in the public-seller setting. Then, among all BNIC and IR mechanisms, the broker's expected profit is maximized by a mechanism in which the object is transferred to the buyer if and only if $\phi_F(v) \ge c$.
\end{theorem}
\begin{proof}
By Theorem \ref{singleton-seller-ic}, we have that $$u_i = \int_a^b \left(\left(v - \frac{1 - F(v)}{F'(v)}\right) - c\right) x(v) \, F'(v) \, dv - u_b(a) - u_s.$$ 
Since we want an IR mechanism $\Mec$ that maximizes $u_i,$ we want $\Mec$ to satisfy the following constraints:
\begin{enumerate}
    \item $u_b(a) = 0$
    \item $u_s = 0$
    \item $x(v) = \begin{cases}
        1 & \text{if $\phi_F(v) \ge c$}\\
        0 & \text{otherwise}
    \end{cases}$
\end{enumerate}
It remains to define $p(v) = (p_s(v), p_b(v))$ such that first two constraints are satisfied. We take $$p_s(v) = x(v)c$$ and $$p_b(v) = x(v) \cdot \inf \{t \ge a | \phi_F(t) \ge c\}.$$ Then, since $p_b(a) = ax(a)$ and $p_s(v) - x(v)c = 0$ for all $v,$ we have that $u_b(a) = 0$ and $u_s = 0,$ as desired.

Note that $\Mec$ is BNIC, since if trade occurs, the buyer pays the lowest valuation he could have reported and still have the trade occur. Then, since $u_b(v)$ is increasing by Theorem \ref{singleton-seller-ic} and $u_b(a) = 0,$ we have that $\Mec$ is IR.
\end{proof}

The next corollary directly follows from the above theorem:
\begin{corollary}
\label{singleton-seller-gft}
    Suppose $F$ is a regular distribution in the public-seller setting. Then, the gains-from-trade from any BNIC and IR mechanism that maximizes the broker's expected profit is $$\int_a^b (v - c) \mathbf{1}\left\{v - c \ge \frac{1 - F(v)}{F'(v)}\right\}\, F'(v) \, dv.$$
\end{corollary}

Using similar ideas, we prove the following theorem which is the public-buyer analogue of Theorem~\ref{thm:publicseller-char}.
Its proof is deferred to Appendix~\ref{proof:singleton-buyer-gft} due to its similarity.
\begin{theorem}
\label{singleton-buyer-gft}
    Suppose $G$ is a regular distribution in the public-buyer setting. Then, the gains-from-trade from any BNIC and IR mechanism that maximizes the broker's expected profit is $$\int_a^b (v - c) \mathbf{1}\left\{v - c \ge \frac{G(c)}{G'(c)}\right\}\, G'(c) \, dc.$$
\end{theorem}

We are finally ready to prove Theorem~\ref{thm:inapx-public}.
% We will now provide a family of problem instances that shows that the ratio between the gains-from-trade of any incentive-compatible, individually-rational mechanism that maximizes the broker's expected profit and the first-best gains-from-trade can be arbitrarily close to $0$ in the public-seller setting.
% \begin{theorem}
%     Let $\epsilon > 0$. Then, there exists a continuous distribution $F$ and a singleton distribution $G$ such that $\GFT < \epsilon\FB,$ where $\GFT$ is the gains-from-trade of an arbitrary incentive-compatible, individually-rational mechanism that maximizes the broker's expected profit and $\FB$ is the first-best gains-from-trade.
% \end{theorem}
\begin{proof}[Proof of Theorem~\ref{thm:inapx-public}]

We first prove the public-seller case.
Let $b > a > 0.$ Define 
\begin{align*}
F(v) = 
\begin{cases} 
    \displaystyle-\frac{a}{v} + 1 &\text{if $a \le v \le b$}\vspace{6pt}\\
    \displaystyle\left(-\frac{a}{b} + 1\right) + \frac{a}{b^2}(v - b) &\text{if $b < v \le 2b$}
\end{cases},
\end{align*}
and let $G$ be the singleton distribution supported on $\{a\}.$ Then, 
\begin{align*}
    \FB &= a\int_a^b (v - a) \cdot \frac{1}{v^2} \, dv + \frac{a}{b^2}\int_b^{2b} (v - a) \, dv\\
    % &= a\left[\ln v + \frac{a}{v}\right]_a^b + \frac{a}{b^2}\left[\frac{v^2}{2} - av\right]_b^{2b}\\
    &= a\left(\ln\left(\frac{b}{a}\right) + \frac{a}{b} - 1\right) + \left(\frac{3a}{2} - \frac{a^2}{b}\right),
\end{align*}
and since $F$ is regular, by Corollary~\ref{singleton-seller-gft} and using the facts that $v = (1-F(v))/F'(v)$ on $(a, b)$ and that $b \le (a+2b)/2 \le 2b$,
\begin{align*}
    \GFT &= \int_a^{2b} (v - a)\mathbf{1}\left\{v - a \ge \frac{1 - F(v)}{F'(v)}\right\} \, F'(v) \, dv\\
    &= \frac{a}{b^2}\int_b^{2b} (v - a)\mathbf{1}\left\{v - a \ge \frac{1 - F(v)}{F'(v)}\right\} \, dv \\
    &= \frac{a}{b^2}\int_b^{2b} (v - a)\mathbf{1}\left\{v \ge \frac{a + 2b}{2}\right\}\, dv\\
    &= \frac{a}{b^2}\int_{(a + 2b) / 2}^{2b} (v - a) \, dv \\
    % &= \frac{a}{b^2}\left[\frac{v^2}{2} - av\right]_{(a + 2b) / 2}^{2b}\\
    &= \frac{3a(a - 2b)^2}{8b^2}.
\end{align*}
Thus, as $b \rightarrow \infty,$ we have that $\displaystyle\frac{\GFT}{\FB} \rightarrow 0.$

Finally, we prove the public-buyer case.
Let $\displaystyle 0 < \delta < \frac{1}{2}.$ Define
\begin{align*}
    G(c) &= \begin{cases}
        \delta c &\text{if $0 \le c \le 1$}\\
        \rho(c) &\text{if $1 < c \le 2$}
    \end{cases},
\end{align*}
where 
\begin{align*}
    \rho(c) &= k - \frac{a}{c - h}
\end{align*}
and
\begin{align*}
    h = \frac{2 - 3\delta}{1 - 2\delta}, \;
    a = \delta(1 - h)^2, \;
    k = \delta(2 - h)
\end{align*}
are chosen so that
\begin{align*}
    \rho(1) = \delta, \;
    \rho'(1) = \delta, \;
    \rho(2) = 1.
\end{align*}
In addition, let $F$ be the singleton distribution supported on $\{2\}.$ Then,
\begin{align*}
    \FB &= \delta\int_0^1 (2 - c) \, dc + a\int_1^2 (2 - c) \cdot \frac{1}{(c - h)^2} \, dc\\
    % &= \delta\left[2c - \frac{c^2}{2}\right]_0^1 + a\left[\frac{h - 2}{c - h} - \ln(c - h)\right]_1^2\\
    &= \frac{3}{2}\delta + a\left(\frac{1}{1 - h} - \ln\left(\frac{2 - h}{1 - h}\right)\right)\\
    &= \frac{3}{2}\delta + \frac{\delta(1 - \delta)^2}{(2\delta - 1)^2}\left(\frac{2\delta - 1}{1 - \delta} - \ln\left(\frac{\delta}{1 - \delta}\right)\right)\\
    &\ge \frac{\delta(1 - \delta)^2}{(2\delta - 1)^2}\left(\frac{2\delta - 1}{1 - \delta} - \ln\left(\frac{\delta}{1 - \delta}\right)\right),
\end{align*}
and since $G$ is regular, by Theorem~\ref{singleton-buyer-gft},
\begin{align*}
    \GFT &= \int_0^2 (2 - c) \mathbf{1}\left\{2 \ge c + \frac{G(c)}{G'(c)}\right\}\, G'(c) \, dc\\
    &= \delta\int_0^1 (2 - c) \, dc\\
    &= \frac{3}{2}\delta,
\end{align*}
since $\displaystyle\phi_G(1) = 1 + \frac{G(1)}{G'(1)} = 2$ and $\phi_G(\cdot)$ is strictly increasing.
Thus,
\begin{align*}
    \frac{\FB}{\GFT} \ge \frac{2(1 - \delta)^2}{3(2\delta - 1)^2}\left(\frac{2\delta - 1}{1 - \delta} - \ln\left(\frac{\delta}{1 - \delta}\right)\right),
\end{align*}
so taking the limit as $\delta \rightarrow 0,$ we obtain
\begin{align*}
    \lim_{\delta \rightarrow 0} \left(\frac{2(1 - \delta)^2}{3(2\delta - 1)^2}\left(\frac{2\delta - 1}{1 - \delta} - \ln\left(\frac{\delta}{1 - \delta}\right)\right)\right) &= \lim_{\delta \rightarrow 0} \frac{2(1 - \delta)}{3(2\delta - 1)} - \lim_{\delta \rightarrow 0} \frac{2(1 - \delta)^2}{3(2\delta - 1)^2} \ln \delta + \lim_{\delta \rightarrow 0} \frac{2(1 - \delta)^2}{3(2\delta - 1)^2} \ln (1 - \delta)\\
    &= -\frac{2}{3} - \lim_{\delta \rightarrow 0} \frac{2(1 - \delta)^2}{3(2\delta - 1)^2} \ln \delta\\
    &= \infty,
\end{align*}
which completes the proof.
\end{proof}

% \subsection{Inapproximability with Public Buyer}
% We now study the case where $G$ is an arbitrary continuous distribution supported on $[a, b]$ and $F$ is a singleton distribution supported on $\{v\}.$ 

% As with the public-seller case, we have the following inapproximability result for the public-buyer case:
% \begin{theorem}
% \label{inapprox-singleton-buy}
%     Let $\epsilon > 0.$ Then, there exists a continuous distribution $G$ and a singleton distribution $F$ such that $\GFT < \epsilon\FB,$ where $\GFT$ is the gains-from-trade of an arbitrary incentive-compatible, individually-rational mechanism that maximizes the broker's expected profit and $\FB$ is the first-best gains-from-trade.
% \end{theorem}
% \begin{proof}

% \end{proof}

\paragraph{Inapproximability with Symmetric Agents}
Finally, we provide the proof of Theorem~\ref{thm:inapx-symmetric} here. The proof mainly follows from considering the same distribution exploited in Theorem~\ref{thm:inapx-general}.
% the inapproximability result under the broker's optimal mechanism persists even in the setting where the buyer's and seller's distributions are identical.

% \sscomment{This can be motivated by size discovery mechanism, see \href{https://theory.stanford.edu/~jvondrak/data/Fixed-price-bilateral.pdf}{link}}
% We now study the case where $F$ and $G$ are identical distributions. 
% $$$$ where $\GFT$ is the gains-from-trade of the mechanism and $\GFT^*$ is the gains-from-trade of any ex-post-efficient mechanism.
\begin{proof}
Let $b > a > 0.$ Let
\begin{align*}
    F(v) &= 
    \begin{cases} 
        \displaystyle-\frac{a}{v} + 1 &\text{if $a \le v \le b$}\vspace{6pt}\\
        \displaystyle\left(-\frac{a}{b} + 1\right) + \frac{a}{b^2}(v - b) &\text{if $b < v \le 2b$},
    \end{cases}
\end{align*}
and let $G = F.$ Then,
\begin{align*}
    \FB &= a^2\int_a^b \int_a^v (v - c) \, \frac{1}{c^2} \, dc \, \frac{1}{v^2} \, dv + \frac{a^2}{b^2}\int_b^{2b} \int_a^b (v - c) \, \frac{1}{c^2} \, dc \, dv + \frac{a^2}{b^4}\int_b^{2b} \int_b^v (v - c) \, dc \, dv\\
    &= a^2\int_a^b \left[-\frac{v}{c} - \ln c\right]_a^v \, \frac{1}{v^2} \, dv + \frac{a^2}{b^2}\int_b^{2b} \left[-\frac{v}{c} - \ln c\right]_a^b \, dv + \frac{a^2}{b^4}\int_b^{2b} \left[vc - \frac{c^2}{2}\right]_b^v \, dv\\
    &= a^2\int_a^b \left(\frac{v}{a} - 1 - \ln \left(\frac{v}{a}\right)\right) \, \frac{1}{v^2} \, dv + \frac{a^2}{b^2}\int_b^{2b} \left(\frac{v}{a} - \frac{v}{b} - \ln \left(\frac{b}{a}\right)\right) dv + \frac{a^2}{2b^4}\int_b^{2b} (v - b)^2 \, dv.
\end{align*}
Expanding the integral once again, we obtain
\begin{align*}
    \FB
    % &= a^2\left[\frac{1}{a} \ln v + \frac{1}{v} + \frac{1}{v}\left(\ln \left(\frac{v}{a}\right) + 1\right) \right]_a^b + \frac{a^2}{b^2}\left[\frac{v^2}{2a} - \frac{v^2}{2b} - v\ln \left(\frac{b}{a}\right)\right]_b^{2b} + \frac{a^2}{2b^4}\left[\frac{(v - b)^3}{3}\right]_b^{2b}\\
    &= a^2\left(\frac{1}{a}\ln\left(\frac{b}{a}\right) + \left(\frac{1}{b} - \frac{1}{a}\right) + \frac{1}{b}\left(\ln\left(\frac{b}{a}\right) + 1\right) - \frac{1}{a}\right) + \frac{a^2}{b^2}\left(\frac{3b^2}{2a} - \frac{3b}{2} - b\ln\left(\frac{b}{a}\right)\right) + \frac{a^2}{2b^4}\left(\frac{b^3}{3}\right)\\
    &= \left(a\ln\left(\frac{b}{a}\right) + \left(\frac{a^2}{b} - a\right) + \frac{a^2}{b}\left(\ln\left(\frac{b}{a}\right) + 1\right) - a\right) + \left(\frac{3a}{2} - \frac{3a^2}{2b} - \frac{a^2}{b}\ln\left(\frac{b}{a}\right)\right) + \frac{a^2}{6b}
    \\
    &\ge
    a \ln (b) - a \ln a - 2a  - \frac{3a^2}{2b}.
\end{align*}
In addition, by Theorem~\ref{thm:myerson}, we have that
\begin{align*}
    \GFT &= \iint (v - c)\mathbf{1}\left\{v - c \ge \frac{1 - F(v)}{F'(v)} + \frac{G(c)}{G'(c)}\right\} \, F'(v) \, dv \, G'(c) \, dc\\
    &= \frac{a}{b^2}\int_b^{2b} \int (v - c)\mathbf{1}\left\{v - c \ge \frac{1 - F(v)}{F'(v)} + \frac{G(c)}{G'(c)}\right\} \, G'(c) \, dc \, dv \\
    &\le \frac{a}{b^2}\int_b^{2b} \int (v - c)\mathbf{1}\left\{v - c \ge 0\right\} \, G'(c) \, dc \, dv\\
    &= \frac{a^2}{b^2}\int_b^{2b} \int_a^b (v - c) \frac{1}{c^2} \, dc \, dv + \frac{a^2}{b^4}\int_b^{2b} \int_b^v (v - c) \, dc \, dv\\
    &= \frac{3a}{2} - \frac{3a^2}{2b} - \frac{a^2}{b}\ln\left(\frac{b}{a}\right) + \frac{a^2}{6b} 
    \\
    &\le
    \frac{3a}{2} + \frac{a^2}{6b}.
\end{align*}
Thus, since $\displaystyle\GFT \sim \frac{3a}{2}$ and $\FB \sim a \ln b,$ we have that $\displaystyle\frac{\GFT}{\FB} \rightarrow 0$ as $b \rightarrow \infty.$
\end{proof}

\gpedit{
\paragraph{Inapproximability of Social Welfare}Interestingly, in several of the settings where we have proven inapproximability of the first-best gains-from-trade, the inapproximability carries over to the first-best social welfare, even though the latter is at least as easy to approximate as the former.

\begin{proof}[Proof of Theorem~\ref{thm:inapx-sw}]
    Suppose we have a family of problem instances $\{(F_\delta, G_\delta)\}_{\delta \in \mathbb{R}_+}$ satisfying the following:
    \begin{enumerate}
        \item The supports of the $G_\delta$ are uniformly upper-bounded by some constant $M.$
        \item For any $(F_\delta, G_\delta),$ we have that $\GFT < \delta \FB.$
        \item $\lim_{\delta \rightarrow 0} \FB = \infty$ 
    \end{enumerate}
    Then, 
    \begin{align*}
        \frac{\SW}{\FBW} &= \frac{\Exu{c \sim G_\delta}{c} + \GFT}{\Exu{c \sim G_\delta}{c} + \FB} \le \frac{M + \GFT}{M + \FB} \le \frac{M + \delta\FB}{M + \FB} = \frac{M / \FB + \delta}{M / \FB + 1}
    \end{align*}
    for any $(F_\delta, G_\delta),$ so since $\displaystyle \frac{M / \FB + \delta}{M / \FB + 1}\rightarrow 0$ as $\delta \rightarrow 0,$ we have that $\displaystyle\frac{\SW}{\FBW} \rightarrow 0$ as $\delta \rightarrow 0.$

    To complete the proof, note that by taking $(F_\delta, G_\delta)$ to be the problem instance from Theorem~\ref{thm:inapx-general} (respectively the public-seller problem instance from Theorem~\ref{thm:inapx-public}) such that $\GFT < \delta\FB,$ we obtain a family of problem instances $\{(F_\delta, G_\delta)\}_{\delta \in \mathbb{R}_+}$ satisfying the three properties above. 
\end{proof}
}

\section{Conclusion}
We study bilateral trade with a broker and identify the intrinsic trade-off between the broker's expected profit and her mechanism's approximation factor to the first-best GFT.
In particular, we characterize the extent to which constant-factor approximation is possible and complement such results by providing inapproximability results beyond those regimes.

Our work introduces several interesting avenues to explore.
First of all, it remains as a major open problem whether a constant-factor approximation is possible for more general classes of distributions (in particular for regular and one-sided MHR distributions) when the broker runs an optimal posted-pricing mechanism.
It would also be interesting to investigate the exact trade-off between the broker's expected profit and the gains-from-trade.
Finally, it would be worth exploring the approximation factor with respect to the social welfare instead of the gains-from-trade, which typically admits more positive results.

% \input{attempt-posted-pricing}

% Paper body

\bibliographystyle{plainnat}
\bibliography{ref}
% \newpage

\appendix
\section{Proof of Theorem~\ref{thm:uniform}}
\label{proof:uniform}
We first prove the lower bounds and then present the worst-case instances that match the lower bounds.
\subsection{Lower bounds}
By Lemma~\ref{lem:normalization}, we can assume that $[a_1, b_1] = [0, 1]$ without loss of generality.
Now, let $F$ and $G$ be the uniform distributions on $[0, 1]$ and $[a, b],$ respectively,
\ie let $$F(v) = \begin{cases} v &\text{if $v \in [0, 1]$}\\ 0 &\text{otherwise}\end{cases}$$ and $$ G(c) = \begin{cases} \displaystyle\frac{c - a}{b - a} &\text{if $c \in [a, b]$}\\ 0 &\text{otherwise}\end{cases}.$$
(Note that $[a, b]$ might not lie completely in $[0, \infty),$ since we normalize $F$ and $G$ so that $F$ is supported on $[0, 1]$.) It is straightforward to check that both $F$ and $G$ are regular. Thus, by Theorem~\ref{thm:myerson}, the gains-from-trade of any BNIC and IR mechanism that maximizes the broker's expected profit is given by
\begin{align*}
    \GFT &= \iint (v - c)\mathbf{1}\left\{v - c \ge \frac{1 - F(v)}{F'(v)} + \frac{G(c)}{G'(c)}\right\} \, dF(v) \, dG(c)\\
    &= \iint (v - c)\mathbf{1}\left\{v - c \ge (1 - v) + (c - a) \right\} \, dF(v) \, dG(c)\\
    &= \iint (v - c)\mathbf{1}\left\{v - c \ge T_a \right\} \, dF(v) \, dG(c),
\end{align*}
where $\displaystyle T_a = \frac{1 - a}{2}$.

\begin{table}
    \centering
    \renewcommand{\arraystretch}{1.5}
    \begin{tabular}{|c|c|c|c|c|c|}
        \hline
        & \multicolumn{5}{c|}{Subcase}\\
        \hline
        Case & $(a)$ & $(b)$ & $(c)$ & $(d)$ & $(e)$\\
        \hline
        1 & $\hyperref[sec:1a]{1}$ & $\hyperref[sec:1b]{2 / 3}$ & $\hyperref[sec:1c]{2 / 3}$ & $\hyperref[sec:1d]{2 / 3}$ &\\
        \hline
        2 & & & $\hyperref[sec:2c]{4 / 7}$ & $\hyperref[sec:2d]{4 / 7}$ & $\hyperref[sec:2e]{1 / 2}$\\
        \hline
        3 & & & $\hyperref[sec:3c]{4 / 7}$ & & $\hyperref[sec:3e]{1 / 2}$\\
        \hline
        4 & & & & $\hyperref[sec:4d]{4 / 7}$ & $\hyperref[sec:4e]{1 / 2}$\\
        \hline
        5 & & & & & $\hyperref[sec:5e]{1 / 2}$\\
        \hline
    \end{tabular}
    \captionof{table}{Summary of Lower Bounds}
    \label{table:uniform-lower-bounds}
\end{table}
We will now complete the proof by considering the following cases based on the values of $a$ and $b$:
\begin{enumerate}[\text{Case} 1:]
    \item $b \le 0$
    \item $a \le 0 < b \le 1$
    \item $a \le 0 < 1 < b$
    \item $0 < a < b \le 1$
    \item $0 < a < 1 < b$
    \item $a \ge 1$
\end{enumerate}

We will further break down each case (except for Case 6) into the following subcases based on the values of $a + T_a$ and $b + T_a:$
\begin{enumerate}[\text{Subcase} (a):]
    \item $b + T_a \le 0$ 
    \item $a + T_a \le 0 < b + T_a \le 1$
    \item $a + T_a \le 0 < 1 < b + T_a$
    \item $0 < a + T_a < b + T_a \le 1$
    \item $0 < a + T_a < 1 < b + T_a$
\end{enumerate}
Note that not every case-subcase combination is possible, as we will see when we investigate each case-subcase combination below. Table~\ref{table:uniform-lower-bounds} summarizes the lower-bounds for each case-subcase combination. (The empty cells correspond to impossible case-subcase combinations.)

We will now derive the lower bounds in Table~\ref{table:uniform-lower-bounds}. The following two tables will be helpful for our analysis. Specifically, Table~\ref{table:case-integrals} summarizes $\FB$ for each case.
\begin{center}
    \renewcommand{\arraystretch}{2}
    \begin{tabular}{|c|c|c|c|}
        \hline
        Case & Case Description & $\FB$ & Derivation\\
        \hline
        1 & $b \le 0$ & $\displaystyle\frac{1 - a - b}{2}$ & $\hyperref[sec:derivations]{1}$\\ 
        \hline
        2 & $a \le 0 < b \le 1$ & $\displaystyle\frac{3a^2 - 3a + b^3 - 3b^2 + 3b}{6(b - a)}$ & $\hyperref[sec:derivations]{2}$\\
        \hline
        3 & $a \le 0 < 1 < b$ & $\displaystyle\frac{3a^2 - 3a + 1}{6(b - a)}$ & $\hyperref[sec:derivations]{3}$\\
        \hline
        4 & $0 < a < b \le 1$ & $\displaystyle\frac{a^2 + ab + b^2 - 3a - 3b + 3}{6}$  & $\hyperref[sec:derivations]{4}$\\
        \hline
        5 & $0 < a < 1 < b$ & $\displaystyle\frac{(1 - a)^3}{6(b - a)}$ & $\hyperref[sec:derivations]{5}$\\
        \hline
    \end{tabular}
    \captionof{table}{First-Best Gains-from-Trade by Case}
    \label{table:case-integrals}
\end{center}
Furthermore, Table~\ref{table:subcase-integrals} summarizes the GFT of the broker's optimal mechanism in each subcase.
\begin{center}
    \renewcommand{\arraystretch}{2}
    \begin{tabular}{|c|c|c|c|}
        \hline
        Subcase & Subcase Description & $\GFT$ & Derivation\\
        \hline
        $(a)$ & $b + T_a \le 0$ & $\displaystyle\frac{1 - a - b}{2}$ & $\hyperref[sec:derivations]{1}$\\ 
        \hline
        $(b)$ & $a + T_a \le 0 < b + T_a \le 1$ & $\displaystyle\frac{(a + b - 1)(a^2 - 4ab + 10a + 4b^2 - 8b + 1)}{24(b - a)}$ & $\hyperref[sec:derivations]{6}$\\
        \hline
        $(c)$ & $a + T_a \le 0 < 1 < b + T_a$ & $\displaystyle\frac{(3a - 1)^2}{24(b - a)}$ & $\hyperref[sec:derivations]{7}$\\
        \hline
        $(d)$ & $0 < a + T_a < b + T_a \le 1$ & $\displaystyle\frac{(a + 2b - 3)^2}{24}$ & $\hyperref[sec:derivations]{8}$\\
        \hline
        $(e)$ & $0 < a + T_a < 1 < b + T_a$ & $\displaystyle \frac{(1 - a)^3}{12(b - a)}$ & $\hyperref[sec:derivations]{9}$\\
        \hline
    \end{tabular}
    \captionof{table}{Gains-from-Trade of Broker's Optimal Mechanism by Subcase}
    \label{table:subcase-integrals}
\end{center}

We now prove the desired lower bounds in what follows.

\paragraph{Case 1: $b \le 0$}

\medskip \noindent \textit{Case 1(a):} $b + T_a \le 0$\\\\
\label{sec:1a}We have that $$\frac{\GFT}{\FB} = 1.$$\\\\
\textit{Case 1(b):} $a + T_a \le 0 < b + T_a \le 1$\\\\
\label{sec:1b}We seek to lower bound $$\frac{\GFT}{\FB} = \frac{-a^2 + 4ab - 10a - 4b^2 + 8b - 1}{12(b - a)}.$$ Taking the derivative with respect to $b,$ we obtain
\begin{align*}
    \frac{\partial}{\partial b}\left[\frac{-a^2 + 4ab - 10a - 4b^2 + 8b - 1}{12(b - a)}\right] &= \frac{12(b - a)(4a - 8b + 8) - 12(-a^2 + 4ab - 10a - 4b^2 + 8b - 1)}{(12(b - a))^2}\\
    &= \frac{(4ab - 8b^2 + 8b - 4a^2 + 8ab - 8a) + (a^2 - 4ab + 10a + 4b^2 - 8b + 1)}{12(b - a)^2}\\
    &= \frac{8ab - 4b^2 - 3a^2 + 2a + 1}{12(b - a)^2}\\
    &= \frac{(-a + 2b + 1)(3a - 2b + 1)}{12(b - a)^2}.
\end{align*}
Now, since $-a + 2b + 1 \ge 0$ and $3a - 2b + 1 \le 0$ whenever $a \in (-\infty, -1]$ and $\displaystyle b \in \left(\frac{a - 1}{2}, \frac{a + 1}{2}\right],$ $\displaystyle\frac{\GFT}{\FB}$ is an decreasing function with respect to $b$ on $a \in (-\infty, -1]$ and $\displaystyle b \in \left(\frac{a - 1}{2}, \frac{a + 1}{2}\right].$ Thus,
\begin{align*}
    \frac{\GFT}{\FB} &= \frac{-a^2 + 4ab - 10a - 4b^2 + 8b - 1}{12(b - a)}\\
    &\ge \frac{-a^2 + (2a^2 + 2a) - 10a - (a^2 + 2a + 1) + (4a + 4) - 1}{6 - 6a} \\
    &= \frac{2 - 6a}{6 - 6a}\\
    &\ge \frac{2 - 6(-1)}{6 - 6(-1)}\\
    &= \frac{2}{3},
\end{align*}
since $\displaystyle\frac{2 - 6a}{6 - 6a}$ is decreasing on $(-\infty, -1]$.
\ \\\\
\textit{Case 1(c):} $a + T_a \le 0 < 1 < b + T_a$\\\\
\label{sec:1c}
Since $(b - a)(1 - a - b)$ is increasing with respect to $b$ on $(-\infty, 0]$, we have that
\begin{align*}
    \frac{\GFT}{\FB} &= \frac{(3a - 1)^2}{12(b - a)(1 - a - b)}\\
    &\ge \frac{(3a - 1)^2}{12a(a - 1)}\\
    &\ge \frac{2}{3},\\ 
\end{align*}
where the last inequality follows from $\displaystyle \frac{(3a - 1)^2}{12a(a - 1)}$ being decreasing on $(-\infty, -1]$.
\ \\
\textit{Case 1(d):} $0 < a + T_a < b + T_a \le 1$\\\\
\label{sec:1d}We seek to lower bound $$\frac{\GFT}{\FB} = \frac{(a + 2b - 3)^2}{12(1 - a - b)}.$$ Taking the derivative with respect to $b,$ we obtain
\begin{align*}
    \frac{\partial}{\partial b}\left[\frac{(a + 2b - 3)^2}{12(1 - a - b)}\right] &= \frac{48(1 - a - b)(a + 2b - 3) + 12(a + 2b - 3)^2}{(12(1 - a - b))^2}\\
    &= \frac{(a + 2b - 3)(1 - 3a - 2b)}{12(1 - a - b)^2}.
\end{align*}
Thus, since $$a + 2b - 3 \le 0$$ and $$1 - 3a - 2b \ge 0$$ whenever $a \in [-1, 0]$ and $b \in [-1, 0],$ we have that $\displaystyle\frac{\GFT}{\FB}$ is decreasing with respect to $b$ on $a \in [-1, 0]$ and $b \in [-1, 0].$ Hence,
\begin{align*}
    \frac{\GFT}{\FB} &= \frac{(a + 2b - 3)^2}{12(1 - a - b)}\\
    &\ge \frac{(a - 3)^2}{12(1 - a)} \\
    &\ge \frac{2}{3},
\end{align*}
since $\displaystyle\frac{(a - 3)^2}{12(1 - a)}$ is increasing on $[-1, 0]$.

\ \\\\
\textit{Case 1(e):} $0 < a + T_a < 1 < b + T_a$\\\\
Since
\begin{align*}
    a + T_a &> 0\\
    a &> -1,
\end{align*}
we have that
\begin{align*}
    b + T_a &> 1\\
    b &> \frac{a + 1}{2} > 0,
\end{align*}
so this case is impossible because $b \le 0$ in Case 1.\\\\
\textbf{Case 2:} $a \le 0 < b \le 1$\\\\
\textit{Case 2(a):} $b + T_a \le 0$\\\\
Since $T_a > 0$ and $b > 0,$ this case is impossible.\\\\
\textit{Case 2(b):} $a + T_a \le 0 < b + T_a \le 1$\\\\
Since 
\begin{align*}
    a + T_a &\le 0\\
    a &\le -1
\end{align*}
and $$b + T_a \le 1,$$ we have that $$b \le \frac{a + 1}{2} \le 0,$$ so this case is impossible because $b > 0$ in Case 2.\\\\
\textit{Case 2(c):} $a + T_a \le 0 < 1 < b + T_a$\\\\
\label{sec:2c}Observe that $b^3 - 3b^2 + 3b$ is increasing on $[0, 1]$ and that $\displaystyle\frac{(3a - 1)^2}{4(3a^2 - 3a + 1)}$ is decreasing on $(-\infty, -1]$.
Thus, we have
\begin{align*}
    \frac{\GFT}{\FB} = \frac{(3a - 1)^2}{4(3a^2 - 3a + b^3 - 3b^2 + 3b)}
    \ge \frac{(3a - 1)^2}{4(3a^2 - 3a + 1)}
    \ge \frac{4}{7} .
\end{align*}
\ \\
\textit{Case 2(d):} $0 < a + T_a < b + T_a \le 1$\\\\
\label{sec:2d}We seek to lower bound $$\frac{\GFT}{\FB} = \frac{(a + 2b - 3)^2(b - a)}{4(3a^2 - 3a + b^3 - 3b^2 + 3b)}.$$ Taking the derivative with respect to $b,$ we obtain
\begin{align*}
    \frac{\partial}{\partial b}\left[\frac{(a + 2b - 3)^2(b - a)}{4(3a^2 - 3a + b^3 - 3b^2 + 3b)}\right] &= \frac{4(3a^2 - 3a + b^3 - 3b^2 + 3b)(4(a + 2b - 3)(b - a) + (a + 2b - 3)^2)}{(4(3a^2 - 3a + b^3 - 3b^2 + 3b))^2}\\
    &- \frac{4(a + 2b - 3)^2(b - a)(3b^2 - 6b + 3)}{(4(3a^2 - 3a + b^3 - 3b^2 + 3b))^2}\\
    &= -\frac{12(3a^2 - 3a + b^3 - 3b^2 + 3b)(a + 2b - 3)(a - 2b + 1)}{(4(3a^2 - 3a + b^3 - 3b^2 + 3b))^2}\\
    &- \frac{12(a + 2b - 3)^2(b - a)(b - 1)^2}{(4(3a^2 - 3a + b^3 - 3b^2 + 3b))^2}\\
    % &= -\frac{12(a + 2b - 3)}{(4(3a^2 - 3a + b^3 - 3b^2 + 3b))^2}((3a^2 - 3a + b^3 - 3b^2 + 3b)(a - 2b + 1)\\ 
    % &+ (a + 2b - 3)(b - a)(b - 1)^2)\\
    &= -\frac{12(a + 2b - 3)}{(4(3a^2 - 3a + b^3 - 3b^2 + 3b))^2}\\
    &\{(3a^3 + ab^3 - 3ab^2 + 9ab - 6a^2b - 2b^4 + 7b^3 - 9b^2 - 3a + 3b)\\
    &+ (-ab^3 + 2b^4 - 7b^3 - a^2b^2 + 5ab^2 + 8b^2 + 2a^2b - 7ab - 3b - a^2 + 3a)\}\\
    &= -\frac{12(a + 2b - 3)}{(4(3a^2 - 3a + b^3 - 3b^2 + 3b))^2}(3a^3 + 2ab^2 + 2ab - 4a^2b - b^2 - a^2b^2 - a^2).
\end{align*}
Now, since each of the terms in the rightmost factor is nonpositive and $a + 2b - 3 \le 0$ whenever $a \in (-1, 0]$ and $\displaystyle b \in \left(0, \frac{a + 1}{2}\right],$ we have that
$\displaystyle\frac{(a + 2b - 3)^2(b - a)}{4(3a^2 - 3a + b^3 - 3b^2 + 3b)}$ is decreasing with respect to $b$ on $a \in (-1, 0]$ and $\displaystyle b \in \left(0, \frac{a + 1}{2}\right].$ Hence,
\begin{align*}
    \frac{\GFT}{\FB} &\ge \frac{(2a - 2)^2((1 - a) / 2)}{4(3a^2 - 3a + (a^3 - 3a^2 + 3a + 7) / 8)} \\
    % &= \frac{4(a - 1)^2(1 - a)}{24a^2 - 24a + (a^3 - 3a^2 + 3a + 7)}\\
    &= \frac{4(1 - a)^3}{a^3 + 21a^2 - 21a + 7}.
\end{align*}
Taking the derivative, we obtain
\begin{align*}
    \frac{d}{da}\left[\frac{4(1 - a)^3}{a^3 + 21a^2 - 21a + 7}\right] &= \frac{-12(1 - a)^2(a^3 + 21a^2 - 21a + 7) - 4(1 - a)^3(3a^2 + 42a - 21)}{(a^3 + 21a^2 - 21a + 7)^2}\\
    &= -\frac{96(1 - a)^2a^2}{(a^3 + 21a^2 - 21a + 7)^2},
\end{align*}
so $\displaystyle\frac{4(1 - a)^3}{a^3 + 21a^2 - 21a + 7}$ is decreasing on $[-1, 0].$ Therefore, $$\frac{\GFT}{\FB} \ge \frac{4}{7}.$$
\ \\
\textit{Case 2(e):} $0 < a + T_a < 1 < b + T_a$\\\\
\label{sec:2e}

Since $b^3 - 3b^2 + 3b$ is increasing on $[0, 1]$ and $\displaystyle\frac{(1 - a)^3}{2(3a^2 - 3a + 1)}$ is decreasing on $[-1, 0]$, we have that
\begin{align*}
    \frac{\GFT}{\FB} = \frac{(1 - a)^3}{2(3a^2 - 3a + b^3 - 3b^2 + 3b)}
    \ge \frac{(1 - a)^3}{2(3a^2 - 3a + 1)}
    \ge \frac{1}{2}. 
\end{align*}
\ \\
\textbf{Case 3:} $a \le 0 < 1 < b$\\\\
\textit{Cases 3(a), (b), (d):} $b + T_a \le 1$\\\\
Since $b > 1$ and $T_a > 0,$ this case is impossible.\\\\
\textit{Case 3(c):} $a + T_a \le 0 < 1 < b + T_a$\\\\
\label{sec:3c}Since $\displaystyle\frac{(3a - 1)^2}{4(3a^2 - 3a + 1)}$ is decreasing on $(-\infty, -1]$, we have that
\begin{align*}
    \frac{\GFT}{\FB} &= \frac{(3a - 1)^2}{4(3a^2 - 3a + 1)}\\
    &\ge \frac{4}{7}.
\end{align*}
\ \\
\textit{Case 3(e):} $0 < a + T_a < 1 < b + T_a$\\\\
\label{sec:3e}Since $\displaystyle\frac{(1 - a)^3}{2(3a^2 - 3a + 1)}$ is decreasing on $[-1, 0]$, we have that
\begin{align*}
    \frac{\GFT}{\FB} &= \frac{(1 - a)^3}{2(3a^2 - 3a + 1)}\\
    &\ge \frac{1}{2}.
\end{align*}
\ \\
\textbf{Case 4:} $0 < a < b \le 1$\\\\
\textit{Case 4(a):} $b + T_a \le 0$\\\\
Since $b > 0$ and $T_a > 0,$ this case is impossible.\\\\
\textit{Cases 4(b), (c):} $a + T_a \le 0$\\\\
Since $a > 0$ and $T_a > 0,$ this case is impossible.\\\\
\textit{Case 4(d):} $0 < a + T_a < b + T_a \le 1$\\\\
\label{sec:4d}We seek to minimize $$\frac{\GFT}{\FB} = \frac{(a + 2b - 3)^2}{4(a^2 + ab + b^2 - 3a - 3b + 3)}.$$ Taking the derivative with respect to $b,$
\begin{align*}
    \frac{\partial}{\partial b}\left[\frac{(a + 2b - 3)^2}{4(a^2 + ab + b^2 - 3a - 3b + 3)}\right] &= \frac{16(a^2 + ab + b^2 - 3a - 3b + 3)(a + 2b - 3) - 4(a + 2b - 3)^3}{(4(a^2 + ab + b^2 - 3a - 3b + 3))^2}\\
    &= \frac{12(a + 2b - 3)(a - 1)^2}{(4(a^2 + ab + b^2 - 3a - 3b + 3))^2}.
\end{align*}
Now, since $a + 2b - 3 \le 0$ on $b \in \displaystyle\left(a, \frac{a + 1}{2}\right],$ we have that $\displaystyle\frac{(a + 2b - 3)^2}{4(a^2 + ab + b^2 - 3a - 3b + 3)}$ is decreasing with respect to $b$ on $a \in (0, 1)$ and $\displaystyle b \in \left(a, \frac{a + 1}{2}\right].$ Thus,
\begin{align*}
    \frac{\GFT}{\FB} &\ge \frac{(2a - 2)^2}{4a^2 + 2a(a + 1) + (a^2 + 2a + 1) - 12a - 6(a + 1) + 12}\\
    &= \frac{(2a - 2)^2}{7(a - 1)^2}\\
    &= \frac{4}{7}.
\end{align*}
\ \\
\textit{Case 4(e):} $0 < a + T_a < 1 < b + T_a$\\\\
\label{sec:4e}We seek to lower bound $$\frac{\GFT}{\FB} = \frac{(1 - a)^3}{2(b - a)(a^2 + ab + b^2 - 3a - 3b + 3)}.$$ Taking the derivative of the denominator with respect to $b,$ we obtain 
\begin{align*}
    \frac{\partial}{\partial b}[2(b - a)(a^2 + ab + b^2 - 3a - 3b + 3)] &= 2(b - a)(2b + a - 3) + 2(a^2 + ab + b^2 - 3a - 3b + 3)\\
    % &= 2(2b^2 - ab - 3b - a^2 + 3a) + 2(a^2 + ab + b^2 - 3a - 3b + 3)\\
    &= 6(b - 1)^2,
\end{align*}
so $2(b - a)(a^2 + ab + b^2 - 3a - 3b + 3)$ is increasing with respect to $b$ on $[0, 1].$ Thus,
\begin{align*}
    \frac{\GFT}{\FB} &= \frac{(1 - a)^3}{2(b - a)(a^2 + ab + b^2 - 3a - 3b + 3)}\\
    &\ge \frac{(1 - a)^3}{2(1 - a)(a^2 - 2a + 1)}\\
    &= \frac{1}{2}.
\end{align*}
\ \\
\textbf{Case 5:} $0 < a < 1 < b$\\\\
\textit{Cases 5(a), (b), (d):} $b + T_a \le 1$\\\\
Since $b > 1$ and $T_a > 0,$ this case is impossible.\\\\
\textit{Case 5(c):} $a + T_a \le 0 < 1 < b + T_a$\\\\
Since $a > 0$ and $T_a > 0,$ this case is impossible.\\\\
\textit{Case 5(e):} $0 < a + T_a < 1 < b + T_a$\\\\
\label{sec:5e}We have that $$\frac{\GFT}{\FB} = \frac{1}{2}.$$\\
\textbf{Case 6:} $a \ge 1$\\\\
We have that $\GFT = 0 = \FB.$

\subsection{Derivations for Table~\ref{table:case-integrals} and~\ref{table:subcase-integrals}}
Here, we provide the derivations of each of the integrals in Tables~\ref{table:case-integrals} and~\ref{table:subcase-integrals}.
Each derivation is indexed by its corresponding number in the tables.
\label{sec:derivations}
\begin{enumerate}
    \item \begin{align*}
    \frac{1}{b - a}\int_0^1 \int_a^b (v - c) \, dc \, dv &= \frac{1}{b - a}\int_0^1 \left[vc - \frac{c^2}{2}\right]_a^b \, dv\\
    &= \frac{1}{b - a}\int_0^1 \left(v(b - a) - \frac{b^2 - a^2}{2}\right) \, dv\\
    % &= \int_0^1 \left(v - \frac{a + b}{2}\right) \, dv\\
    % &= \left[\frac{v^2}{2} - \frac{(a + b)v}{2}\right]_0^1\\
    &= \frac{1 - a - b}{2}
    \end{align*}
    \item \begin{align*}
        &\frac{1}{b - a}\int_0^b \int_a^v (v - c) \, dc \, dv + \frac{1}{b - a}\int_b^1 \int_a^b (v - c) \, dc \, dv\\
        &= \frac{1}{b - a}\int_0^b \left[vc - \frac{c^2}{2}\right]_a^v \, dv + \frac{1}{b - a}\int_b^1 \left[vc - \frac{c^2}{2}\right]_a^b \, dv\\
        % &= \frac{1}{b - a}\int_0^b \frac{(v - a)^2}{2} \, dv + \frac{1}{b - a}\int_b^1 \left(v(b - a) - \frac{b^2 - a^2}{2}\right) \, dv\\
        % &= \frac{1}{b - a}\left[\frac{(v - a)^3}{6}\right]_0^b + \frac{1}{b - a}\left[\frac{(b - a)v^2}{2} - \left(\frac{b^2 - a^2}{2}\right)v\right]_b^1\\
        % &= \frac{(b - a)^3 + a^3}{6(b - a)} + \frac{(b - a)(1 - b)(1 - a)}{2(b - a)}\\
        &= \frac{3a^2 - 3a + b^3 - 3b^2 + 3b}{6(b - a)}
        \end{align*}
    \item \begin{align*}
        \frac{1}{b - a}\int_0^1 \int_a^v (v - c) \, dc \, dv
        &= \frac{1}{b - a}\int_0^1 \left[vc - \frac{c^2}{2}\right]_a^v \, dv\\
        % &= \frac{1}{b - a}\int_0^1 \left(v(v - a) - \frac{v^2 - a^2}{2}\right) \, dv\\
        % &= \frac{1}{b - a}\int_0^1 \frac{(v - a)^2}{2} \, dv\\
        % &= \frac{1}{b - a}\left[\frac{(v - a)^3}{6}\right]_0^1\\
        &= \frac{3a^2 - 3a + 1}{6(b - a)}
        \end{align*}
    \item \begin{align*}
        \frac{1}{b - a}\int_a^b \int_c^1 (v - c) \, dv \, dc &= \frac{1}{b - a}\int_a^b \left[\frac{v^2}{2} - cv\right]_c^1 \, dc\\
        &= \frac{1}{b - a}\int_a^b \left[\frac{1 - c^2}{2} - c(1 - c)\right] \, dc\\
        % &= \frac{1}{b - a}\int_a^b \frac{(c - 1)^2}{2}\, dc\\
        % &= \frac{1}{b - a}\left[\frac{(c - 1)^3}{6}\right]_a^b\\
        % &= \frac{(b^3 - a^3) - 3(b^2 - a^2) + 3(b - a)}{6(b - a)}\\
        &= \frac{a^2 + ab + b^2 - 3a - 3b + 3}{6}\\
        \end{align*}
    \item \begin{align*}
        \frac{1}{b - a}\int_a^1 \int_a^v (v - c) \, dc \, dv &= \frac{1}{b - a}\int_a^1 \left[vc - \frac{c^2}{2}\right]_a^v \, dv\\
        &= \frac{1}{b - a}\int_a^1 \left(v(v - a) - \frac{v^2 - a^2}{2}\right) \, dv\\
        % &= \frac{1}{b - a}\int_a^1 \frac{(v - a)^2}{2} \, dv\\
        % &= \frac{1}{b - a}\left[\frac{(v - a)^3}{6}\right]_a^1\\
        &= \frac{(1 - a)^3}{6(b - a)}
        \end{align*}
    \item \begin{align*}
        &\frac{1}{b - a}\int_0^{b + T_a} \int_a^{v - T_a} (v - c) \, dc \, dv + \frac{1}{b - a}\int_{b + T_a}^1 \int_a^b (v - c) \, dc \, dv\\ 
        % &= \frac{1}{b - a}\int_0^{b + T_a} \left[vc - \frac{c^2}{2}\right]_a^{v - T_a} \, dv + \frac{1}{b - a}\int_{b + T_a}^1 \left[vc - \frac{c^2}{2}\right]_a^b \, dv\\
        &= \frac{1}{b - a}\int_0^{b + T_a} \left(v(v - T_a - a) - \frac{(v - T_a)^2 - a^2}{2}\right) \, dv + \frac{1}{b - a}\int_{b + T_a}^1 \left(v(b - a) - \frac{b^2 - a^2}{2}\right) \, dv\\ 
        % &= \frac{1}{b - a}\int_0^{b + T_a} \left(\frac{v^2}{2} - av - \frac{T_a^2}{2} + \frac{a^2}{2}\right) \, dv + \int_{b + T_a}^1 \left(v - \frac{a + b}{2}\right) \, dv\\ 
        &= \frac{1}{b - a}\int_0^{b + T_a} \left(\frac{v^2}{2} - av + \frac{3a^2 + 2a - 1}{8}\right) \, dv + \int_{b + T_a}^1 \left(v - \frac{a + b}{2}\right) \, dv\\
        % &= \frac{1}{b - a}\left[\frac{v^3}{6} - \frac{av^2}{2} + \frac{(3a^2 + 2a - 1)v}{8}\right]_0^{b + T_a} + \left[\frac{v^2}{2} - \frac{(a + b)v}{2}\right]_{b + T_a}^1\\
        &= \frac{b + T_a}{b - a}\left(\frac{(b + T_a)^2}{6} - \frac{a(b + T_a)}{2} + \frac{3a^2 + 2a - 1}{8}\right) + \left(\frac{1 - (b + T_a)^2}{2} - \frac{(a + b)(1 - b - T_a)}{2}\right)\\
        % &= \frac{(b + T_a)(4b^2 + 8bT_a + 4T_a^2 - 12ab - 12aT_a + 9a^2 + 6a - 3)}{24(b - a)} - \frac{(T_a + b - 1)(T_a - a + 1)}{2}\\
        % &= \frac{(2b - a + 1)(2b^2 + 2b - 8ab - 1 - a + 8a^2)}{24(b - a)} - \frac{9(a - 2b + 1)(a - 1)(b - a)}{24(b - a)}\\
        % &= \frac{4b^3 + 6b^2 - 18ab^2 - 12ab + 24a^2b + 9a^2 - 8a^3 - 1}{24(b - a)} - \frac{27a^2b - 9a^3 - 18ab - 18ab^2 + 18b^2 - 9b + 9a}{24(b - a)}\\
        % &= \frac{4b^3 - 12b^2 + 6ab - 3a^2b + 9a^2 + a^3 - 1 + 9b - 9a}{24(b - a)}\\
        &= \frac{(a + b - 1)(a^2 - 4ab + 10a + 4b^2 - 8b + 1)}{24(b - a)}
        \end{align*}
    \item \begin{align*}
        \frac{1}{b - a}\int_0^1 \int_a^{v - T_a} (v - c) \, dc \, dv &= \frac{1}{b - a}\int_0^1 \left[vc - \frac{c^2}{2}\right]_a^{v - T_a} \, dv\\ 
        % &= \frac{1}{b - a}\int_0^1 \left(v(v - T_a - a) - \frac{(v - T_a)^2 - a^2}{2}\right) \, dv\\
        &= \frac{1}{b - a}\int_0^1 \left(\frac{v^2}{2} - av - \frac{T_a^2}{2} + \frac{a^2}{2}\right) \, dv\\
        &= \frac{1}{b - a}\int_0^1 \left(\frac{v^2}{2} - av + \frac{3a^2 + 2a - 1}{8}\right) \, dv\\
        % &= \frac{1}{b - a}\left[\frac{v^3}{6} - \frac{av^2}{2} + \frac{(3a^2 + 2a - 1)v}{8}\right]_0^1\\
        % &= \frac{1}{b - a}\left(\frac{1}{6} - \frac{a}{2} + \frac{3a^2 + 2a - 1}{8}\right)\\
        &= \frac{(3a - 1)^2}{24(b - a)}
    \end{align*}
    \item \begin{align*}
        \frac{1}{b - a}\int_a^b \int_{c + T_a}^1 (v - c) \, dv \, dc 
        % &= \frac{1}{b - a}\int_a^b \left[\frac{v^2}{2} - cv\right]_{c + T_a}^1 \, dc\\ 
        &= \frac{1}{b - a}\int_a^b \left(\frac{1 - (c + T_a)^2}{2} - c(1 - c - T_a)\right) \, dc\\ 
        &= \frac{1}{b - a}\int_a^b \left(\frac{c^2}{2} - c + \frac{1 - T_a^2}{2}\right) \, dc\\
        % &= \frac{1}{b - a}\left[\frac{c^3}{6} - \frac{c^2}{2} + \frac{(1 - T_a^2)c}{2}\right]_a^b\\
        &= \frac{1}{b - a}\left(\frac{b^3 - a^3}{6} - \frac{b^2 - a^2}{2} + \frac{(1 - T_a^2)(b - a)}{2}\right)\\
        &= \frac{a^2 + ab + b^2}{6} - \frac{a + b}{2} + \frac{3 + 2a - a^2}{8}\\
        &= \frac{(a + 2b - 3)^2}{24}
    \end{align*}
    \item \begin{align*}
        \frac{1}{b - a}\int_{a + T_a}^1 \int_a^{v - T_a} (v - c) \, dc \, dv
        % &= \frac{1}{b - a}\int_{a + T_a}^1 \left[vc - \frac{c^2}{2}\right]_a^{v - T_a} \, dv\\
        % &= \frac{1}{b - a}\int_{a + T_a}^1 \left(v(v - T_a - a) - \frac{(v - T_a)^2 - a^2}{2}\right) \, dv\\
        &= \frac{1}{b - a}\int_{a + T_a}^1 \left(\frac{v^2}{2} - av - \frac{T_a^2}{2} + \frac{a^2}{2}\right) \, dv\\
        &= \frac{1}{b - a}\int_{a + T_a}^1 \left(\frac{v^2}{2} - av + \frac{3a^2 + 2a - 1}{8}\right) \, dv\\
        % &= \frac{1}{b - a}\left[\frac{v^3}{6} - \frac{av^2}{2} + \frac{(3a^2 + 2a - 1)v}{8}\right]_{a + T_a}^1\\
        % &= \frac{1}{b - a}\left(\frac{1}{6} - \frac{a}{2} + \frac{3a^2 + 2a - 1}{8}\right)\\ 
        % &- \frac{a + T_a}{b - a}\left(\frac{(a + T_a)^2}{6} - \frac{a(a + T_a)}{2} + \frac{3a^2 + 2a - 1}{8}\right)\\
        &= \frac{(3a - 1)^2}{24(b - a)} - \frac{(a + T_a)(a^2 - 4aT_a + 4T_a^2 + 6a - 3)}{24(b - a)}\\
        &= \frac{(3a - 1)^2}{24(b - a)} - \frac{(a + 1)(2a^2 + a - 1)}{24(b - a)}\\
        &= \frac{(1 - a)^3}{12(b - a)}\\
        \end{align*}
\end{enumerate}

\subsection{Upper Bounds}
Table~\ref{table:uniform-tight-bound-examples} summarizes the specific problem instances that yield the tightness results.
In particular, $1 / 2$ is a tight lower bound on the approximation factor to the first-best GFT under the broker's optimal mechanism in the uniform setting.
We omit the calculations as it is a simple plugging-in of the described values.
\begin{table}
    \centering
    \renewcommand{\arraystretch}{1.5}
    \begin{tabular}{|c|c|c|c|c|c|}
        \hline
        & \multicolumn{5}{c|}{Subcase}\\
        \hline
        Case & $(a)$ & $(b)$ & $(c)$ & $(d)$ & $(e)$\\
        \hline
        1 & $[-3, -2]$ & $[-1, 0]$ & $[-1 - \epsilon, 0]$ & $[-1 + \epsilon, 0]$ &\\
        \hline
        2 & & & $[-1, 1]$ & $[0, 1/2]$ & $[0, 1]$\\
        \hline
        3 & & & $[-1, 2]$ & & $[0, 2]$\\
        \hline
        4 & & & & $[1 / 3, 2 / 3]$ & $[1 / 2, 1]$\\
        \hline
        5 & & & & & $[1 / 2, 3 / 2]$\\
        \hline
    \end{tabular}
    \captionof{table}{Examples of Values of $a$ and $b$ that Match Lower Bounds}\label{table:uniform-tight-bound-examples}
\end{table}

\section{Proof of Theorem~\ref{singleton-buyer-gft}; Public-Buyer}
\label{proof:singleton-buyer-gft}
We first introduce some notation:
\begin{definition}
    Given a direct mechanism $\Mec = (x(c), p(c))$\footnote{In the public-buyer setting, the allocation function $x$ and payment function $p$ of $\Mec$ only depend on the seller's reported valuation $c.$}, we define the following functions for the expected utilities of the buyer, intermediary, and seller, respectively, assuming truthful reporting by the seller:
    \begin{itemize}
        \item $\displaystyle u_b = \int_a^b (x(c) \cdot v - p_b(c)) \, G'(c) \, dc$ 
        \item $\displaystyle u_i = \int_a^b (p_b(c) - p_s(c)) \, G'(c) \, dc$
        \item $\displaystyle u_s(c) = p_s(c) - x(c) \cdot c$
    \end{itemize}
\end{definition}
Note that with a public buyer, $\Mec$ is BNIC if $u_s(c) \ge p_s(\tilde{c}) - x(\tilde{c}) \cdot c$ for every $c$ and $\tilde{c},$ and $\Mec$ is IR if $u_b \ge 0$ and $u_s(c) \ge 0$ for all $c.$

The following theorem is an analogue of Theorem 3 in~\cite{myerson1983efficient}:
\begin{theorem}
\label{singleton-buyer-ic}
For any BNIC mechanism $\Mec$ in the public-buyer setting, $u_s(c)$ is decreasing and $$u_i + u_b + u_s(b) = \int_a^b \left(v - \left(c + \frac{G(c)}{G'(c)}\right)\right) x(c) \, G'(c) \, dc.$$
\end{theorem}
\begin{proof}
Since $\Mec$ is BNIC, for every $c$ and $c',$ $$u_s(c) = p_s(c) - x(c) \cdot c \ge p_s(c') - x(c') \cdot c$$ and $$u_s(c') = p_s(c') - x(c') \cdot c' \ge p_s(c) - x(c) \cdot c'.$$ Thus, $$(c' - c)x(c) \ge u_s(c) - u_s(c') \ge (c' - c)x(c').$$ In particular, if $c' > c,$ then $x(c) \ge x(c'),$ so $x(c)$ is decreasing and hence integrable. Therefore, $u'_s(c) = -x(c)$ almost everywhere, so 
\begin{equation}
u_s(c) = u_s(b) + \int_v^b x(t) \, dt,\label{eq:u_s_integral}
\end{equation}
and $u_s(c)$ is decreasing.\\\\
Now, by (\ref{eq:u_s_integral}), we have that
\begin{align*}
    \int_a^b (v - c)x(c) \, G'(c) \, dc - u_i &= \int_a^b (vx(c) - p_b(c)) \, G'(c) \, dc + \int_a^b (p_s(c) - cx(c)) \, G'(c) \, dc\\
    &= u_b + \int_a^b u_s(c) \, G'(c) \, dc\\
    &= u_b + \int_a^b \left(u_s(b) + \int_v^b x(t) \, dt\right) \, G'(c) \, dc \\
    &= u_b + u_s(b) + \int_a^b \int_v^b x(t) \, dt \, G'(c) \, dc \\
    &= u_b + u_s(b) + \int_a^b \int_a^t  G'(c) \, dc \, x(t) \, dt \\
    &= u_b + u_s(b) + \int_a^b G(t) \, x(t) \, dt.
\end{align*}
Finally, from the first and last expressions in this chain of equations, we obtain $$u_i + u_b + u_s(b) = \int_a^b \left(v - \left(c + \frac{G(c)}{G'(c)}\right)\right) x(c) \, G'(c) \, dc.$$
\end{proof}

Next, the following is an analogue of Theorem 4 in~\cite{myerson1983efficient}, tailored to the public-buyer setting:
\begin{theorem}
    Suppose $G$ is a regular distribution in the public-buyer setting. Then, among all BNIC and IR mechanisms, the broker's expected profit is maximized by a mechanism in which the object is transferred to the buyer if and only if $\phi_G(c) \le v$.
\end{theorem}
\begin{proof}
By Theorem \ref{singleton-buyer-ic}, we have that $$u_i = \int_a^b \left(v - \left(c + \frac{G(c)}{G'(c)}\right)\right) x(c) \, G'(c) - u_b - u_s(b).$$ Since we want an IR mechanism $\Mec$ that maximizes $u_i,$ we want $\Mec$ to satisfy the following constraints:
\begin{enumerate}
    \item $u_b = 0$
    \item $u_s(b) = 0$
    \item $x(c) = \begin{cases}
        1 & \text{if $\phi_G(c) \le v$}\\
        0 & \text{otherwise}
    \end{cases}$
\end{enumerate}
It remains to define $p(c) = (p_s(c), p_b(c))$ such that first two constraints are satisfied. We take $$p_s(c) = x(c) \cdot \sup \{t \le b | \phi_G(t) \le v\}$$ and $$p_b(c) = x(c)v.$$ Then, since $x(c)v - p_b(c) = 0$ for all $c$ and $p_s(b) = x(b)b,$ we have that $u_b = 0$ and $u_s(b) = 0,$ as desired.\\\\
Note that $\Mec$ is BNIC, since if trade occurs, the seller pays the highest valuation he could have reported and still have the trade occur. Then, since $u_s(c)$ is decreasing by Theorem \ref{singleton-buyer-ic} and $u_s(b) = 0,$ we have that $\Mec$ is IR.
\end{proof}

Theorem \ref{singleton-buyer-gft} directly follows from the above theorem.
\end{document}